\documentclass[11pt,twoside]{extarticle}

\usepackage[papersize={8.5in,11in},margin=1in,top=1in,headsep=5mm]{geometry}
\usepackage[utf8]{inputenc}
\usepackage[bottom]{footmisc}

\usepackage[dvipsnames]{xcolor}
\usepackage{fancyhdr}
\date{}
\fancypagestyle{plain}{
  \fancyhead[C,R]{}
  \fancyfoot{} % comment to put page number at bottom
}
\fancypagestyle{plain}{
  \fancyhead[C,R]{}
  \fancyfoot{} % comment to put page number at bottom
  }
\usepackage[dvipsnames]{xcolor}
\usepackage[colorlinks,linkcolor = Orchid, urlcolor  = Peach, citecolor = SeaGreen, anchorcolor = Peach]{hyperref}
\usepackage{times}
\usepackage[textwidth=2.8cm,colorinlistoftodos, shadow,color=blue!30!white]{todonotes}
\usepackage{caption}
\usepackage[noend, boxruled, linesnumbered] {algorithm2e}
\usepackage{algorithmic}

\usepackage[utf8]{inputenc} % allow utf-8 input
\usepackage[T1]{fontenc}    % use 8-bit T1 fonts
\usepackage{hyperref}       % hyperlinks
\usepackage{url}            % simple URL typesetting
\usepackage{booktabs}       % professional-quality tables
\usepackage{amsfonts}       % blackboard math symbols
\usepackage{nicefrac}       % compact symbols for 1/2, etc.
\usepackage{microtype}      % microtypography
\usepackage{amssymb}
\usepackage{amsmath}
\usepackage{tabu}
\usepackage{caption}
\usepackage{multirow}
\usepackage{nicefrac}
\usepackage{comment}
\usepackage{color}
\usepackage{array}
\newcolumntype{P}[1]{>{\centering\arraybackslash}m{#1}}

\DeclareMathOperator*{\argmin}{argmin}
\DeclareMathOperator*{\argmax}{argmax}
\newcommand{\R}{\mathbb{R}}
\newcommand{\N}{\mathbb{N}}
\DeclareMathOperator*{\E}{\mathbb{E}}
\newcommand{\norm}[1]{\|#1\|}
\newcommand{\normFull}[1]{\left\| #1 \right\|}
\newcommand{\mat}[1]{\boldsymbol{#1}}
\newcommand{\upperthres}[2]{\Pi_{#2}(#1)}
\newcommand{\lowerthres}[2]{\bar{\Pi}_{#2}(#1)}
\newcommand{\topind}[2]{I_{#2}(#1)}
\newcommand{\botind}[2]{\bar{I}_{#2}(#1)}
\newcommand{\gap}{\mathrm{gap}}
\newcommand{\sr}{\mathrm{sr}}
\newcommand{\nnz}{\mathrm{nnz}}
\newcommand{\defeq}{\stackrel{\mathrm{\scriptscriptstyle def}}{=}}
\newcommand{\otilde}{\tilde{O}}
\usepackage{thmtools}
\usepackage{thm-restate}

\usepackage{hyperref}

\usepackage{cleveref}

\newtheorem{theorem}{Theorem}
\newtheorem{lemma}[theorem]{Lemma}

\newtheorem{definition}[theorem]{Definition}

\newenvironment{proof}%
{%
 \par\noindent{\bfseries\upshape Proof\ }%
}%

\title{Exploiting Numerical Sparsity for Efficient Learning : \\
Faster Eigenvector Computation and Regression}

\author{
  Neha Gupta \\
  Stanford University\\
  \texttt{nehagupta@cs.stanford.edu} \\
   \and
   Aaron Sidford \\
   Stanford University \\
   \texttt{sidford@stanford.edu} \\
}
\begin{document}

\maketitle

\begin{abstract}
In this paper, we obtain improved running times for regression and top eigenvector computation for numerically sparse matrices. Given a data matrix $\mat{A} \in \R^{n \times d}$ where every row $a \in \R^d$ has $\|a\|_2^2 \leq L$ and numerical sparsity at most $s$, i.e. $\|a\|_1^2 / \|a\|_2^2 \leq s$, we provide faster algorithms for these problems in many parameter settings.

For top eigenvector computation, we obtain a running time of $\otilde(nd + r(s + \sqrt{r s}) / \gap^2)$ where $\gap > 0$ is the relative gap between the top two eigenvectors of $\mat{A}^\top \mat{A}$ and $r$ is the stable rank of $\mat{A}$. This running time improves upon the previous best unaccelerated running time of $O(nd + r d / \gap^2)$ as it is always the case that $r \leq d$ and $s \leq d$.

For regression, we obtain a running time of $\otilde(nd + (nL / \mu) \sqrt{s nL / \mu})$ where $\mu > 0$ is the smallest eigenvalue of $\mat{A}^\top \mat{A}$. This running time improves upon the previous best unaccelerated running time of $\otilde(nd + n L d / \mu)$. This result expands the regimes where regression can be solved in nearly linear time from when $L/\mu = \otilde(1)$ to when $L / \mu = \otilde(d^{2/3} / (sn)^{1/3})$.

Furthermore, we obtain similar improvements even when row norms and numerical sparsities are non-uniform and we show how to achieve even faster running times by accelerating using approximate proximal point \cite{frostig2015regularizing} / catalyst \cite{lin2015universal}. Our running times depend only on the size of the input and natural numerical measures of the matrix, i.e. eigenvalues and $\ell_p$ norms, making progress on a key open problem regarding optimal running times for efficient large-scale learning.

  \end{abstract}
  % !TEX root = main.tex
\section{Introduction}
Regression and top eigenvector computation are two of the most fundamental problems in learning, optimization, and numerical linear algebra. They are central tools for data analysis and of the simplest problems in a hierarchy of complex machine learning computational problems. Consequently, developing provably faster algorithms for these problems is often a first step towards deriving new theoretically motivated algorithms for large scale data analysis.

Both regression and top eigenvector computation are known to be efficiently reducible \cite{garber2016faster} to the more general and prevalent
finite sum optimization problem of minimizing a convex function $f$ decomposed
into the sum of $m$ functions $f_1, ... , f_m$, i.e.
$\min_{x \in \R^n} f(x)
\text{ where }
f(x) = \frac{1}{m} \sum_{i \in [m]} f_i(x).$
This optimization problem encapsulates a variety of learning tasks where we have data points $\{(a_1,b_1), (a_2,b_2), \cdots, (a_n,b_n)\}$ corresponding to feature vectors $a_i$, labels $b_i$, and we wish to find the predictor $x$ that minimizes the average loss of predicting $b_i$ from $a_i$ using $x$, denoted by $f_i(x)$.

Given the centrality of this problem to machine learning and optimization, over the past few years there has been extensive research on designing new provably efficient methods for solving this problem \cite{johnson2013accelerating, frostig2015regularizing, konevcny2017semi, bottou2004large, shalev2013stochastic}. Using a variety of sampling techniques, impressive running time improvements have been achieved. The emphasis in this line of work has been on improving the dependence on the number of gradient evaluations of the $f_i$ that need to be performed, i.e. improving dependence on $m$, as well as improving the dependence on other problem parameters.

The question of what structural assumptions on $f_i$ allow even faster running times to be achieved, is much less studied. A natural and fundamental question in this space, is when can we achieve faster running times by computing the gradients of $f_i$ approximately, thereby decreasing iteration costs. While there has been work on combining coordinate descent methods with these stochastic methods \cite{konevcny2017semi}, in the simple cases of regression and top eigenvector computation these methods do not yield any improvement in iteration cost. More broadly, we are unaware of previous work on linearly convergent algorithms with faster running times for finite sum problems through this approach.\footnote{During the final editing of this paper, we found that there was an independent result \cite{tan2018stochastic} which also investigates the power of iterative optimization algorithms with $o(d)$ iteration costs and we leave it to further research to compare these results.}

In this paper, we advance our understanding of the computational power of subsampling gradients of the $f_i$ for the problems of top eigenvector computation and regression. In particular, we show that under assumptions of numerical sparsity of the input matrix, we can achieve provably faster algorithms and new nearly linear time algorithms for a broad range of parameters. We achieve our result by applying coordinate sampling techniques to Stochastic Variance Reduced Gradient Descent (SVRG) \cite{konevcny2017semi, johnson2013accelerating}, a popular tool for finite sum optimization, along with linear algebraic data structures (in the case of eigenvector computation) that we believe may be of independent interest.

The results in this paper constitute an important step towards resolving a key gap in our understanding of optimal iterative methods for top eigenvector computation and regression. Ideally, running times of these problems would depend only on the size of the input, e.g. the number of non-zero entries in the input data matrix, row norms, eigenvalues, etc. However, this is not the case for the current fastest regression algorithms as these methods work by
picking rows of the matrix non-uniformly yielding expected iteration costs that depend on brittle weighted sparsity measures (which for simplicity are typically instead stated in terms of the maximum sparsity among all rows, See Section~\ref{sec:prev_algs}). This causes particularly unusual running times for related problems like nuclear norm estimation \cite{musco2017spectrum}.

This paper takes an important step towards resolving this problem by providing running times for top eigenvector computation and regression that depend only on the size of the input and natural numerical quantities like eigenvalues, $\ell_1$-norms, $\ell_2$-norms, etc. While our running times do not strictly dominate those based on the sparsity structure of the input (and it is unclear if such running times are possible), they improve upon the previous work in many settings.  Ultimately, we hope this paper provides useful tools for even faster algorithms for solving large scale learning problems.
\subsection{The Problems}
\label{section:problems}
Throughout this paper we let $\mat{A} \in \R^{n \times d}$ denote a data matrix with rows $a_1, ... , a_n \in \R^{d}$. We let  $\sr(\mat{A}) \defeq  \norm{\mat{A}}_F^2 / \norm{\mat{A}}_2^2$ denote the stable rank of $\mat{A}$ and we let $\nnz(\mat{A})$ denote the number of non-zero entries in $\mat{A}$. For symmetric $\mat{M} \in \R^{d \times d}$ we let $\lambda_{1}(\mat{M}) \geq \lambda_2(\mat{M}) \geq ... \geq \lambda_{d}(\mat{M})$ denote its eigenvalues, $\norm{x}^2_{\mat{M}} = x^\top \mat{M}x$ and we let $\gap(\mat{M}) \defeq (\lambda_1(\mat{M}) - \lambda_2(\mat{M})) / \lambda_1(\mat{M})$ denote its (relative) eigenvalue gap. For convenience we let $\gap \defeq \gap(\mat{A}^\top \mat{A})$, $\lambda_1 \defeq \lambda_1(\mat{A}^\top \mat{A})$, $\mu \defeq \lambda_{\min} \defeq \lambda_d(\mat{A}^\top \mat{A})$, $\sr \defeq \sr(\mat{A})$, $\nnz(\mat{A}) \defeq \nnz$, $\kappa \defeq \norm{\mat{A}}_F^2 / \mu$, and $\kappa_{\max} \defeq \lambda / \mu$. We use $\tilde{O}$ notation to hide polylogarithmic factors in the input parameters and error rates. With this notation, we consider the following two optimization problems.
\begin{definition} [\textbf{Top Eigenvector Problem}] \label{problem:eigenvector}
Find $v^* \in \R^d$ such that
\[
v^* = \argmax_{x \in \R^d, \norm{x}_2 = 1} x^\top \mat{A}^\top \mat{A} x
\]
We call $v$ an \emph{$\epsilon$-approximate solution} to the problem if $\norm{v}_2 = 1$ and
\(
v^\top\mat{A}^\top\mat{A}v \geq (1 - \epsilon)\lambda_1(\mat{A}^\top\mat{A}) ~.
\)
\end{definition}
\begin{definition}[\textbf{Regression Problem}]  \label{problem:regression}
Given $b \in \R^n$ find $x^* \in \R^d$ such that
\[x^* = \argmin_{x \in \R^d}\norm{\mat{A}x - b}_2^2\]
Given initial $x_0 \in \R^d$, we call $x$  an \emph{$\epsilon$-approximate solution} if
$\norm{x - x^*}_{\mat{A}^\top\mat{A}} \leq \epsilon\norm{x_0 - x^*}_{\mat{A}^\top\mat{A}}$.
\end{definition}
Each of these are known to be reducible to the finite sum optimization problem. The regression problem is equivalent to the finite sum problem with $f_i(x) \defeq (m/2) (a_i^\top x - b_i)^2$ and the top eigenvector problem is reducible with only polylogarithmic overhead to the finite sum problem with
$f_i(x) \defeq (\lambda/2) \norm{x - x_0}_2^2 - (m/2) (a_i^\top (x - x_0))^2 + b_i^\top x$ for carefully chosen
$\lambda$ and $x_0$, and some $b_is$ \cite{garber2016faster}.
\subsection{Our Results}
In this paper, we provide improved iterative methods for top eigenvector computation and regression that depend only on regularity parameters and not the specific sparsity structure of the input. Rather than assuming uniform row sparsity as in previous work, our running times depend on the numerical sparsity of rows of $\mat{A}$, i.e. $s_i \defeq \norm{a_i}_1^2 / \norm{a_i}_2^2$, which is at most the row sparsity, but may be smaller.

Note that our results, as stated, are worse as compared to the previous running times which depend on the $\ell_0$ sparsity
 in some parameter regimes. For simplicity, we state our results in terms of only the numerical sparsity.
 However, when the number of non-zero entries in a row is small, we can always choose that row completely and not sample it, yielding results that are always as good as the previous results and strictly better in some parameter regimes.
\subsubsection{Top Eigenvector Computation}
For top eigenvector computation, we give an unaccelerated running time of $$\tilde{O}\left(\nnz(\mat{A}) + \frac{1}{\gap^2}\sum_i \frac{\norm{a_i}_2^2}{\lambda_1}\left(\sqrt{s_i} + \sqrt{\sr(\mat{A})}\right)\sqrt{s_i}\right)$$
and an accelerated running time of $$\tilde{O}\left(\nnz(\mat{A}) + \frac{\nnz(\mat{A})^{\frac{3}{4}}}{\sqrt{\gap}}\left({\sum_i \frac{\norm{a_i}_2^2}{\lambda_1}\left(\sqrt{s_i} + \sqrt{\sr(\mat{A})}\right)\sqrt{s_i})}\right)^{\frac{1}{4}}\right)$$
 as compared to the previous unaccelerated running time
of $$\tilde{O}\left(\nnz(\mat{A}) + \max_i\nnz(a_i) \frac{\sr(\mat{A})}{\gap^2}\right)$$ and accelerated iterative methods of
$$\tilde{O}\left(\nnz(\mat{A}) + \nnz(\mat{A})^{\frac{3}{4}}\left(\max_i\nnz(a_i) \frac{\sr(\mat{A})}{\gap^2}\right)^{\frac{1}{4}}\right) ~.$$

 In the simpler case of uniform
uniform row norms $\norm{a_i}_2^2 = \norm{a}_2^2$ and uniform row sparsity $s_i = s$, our running time (unaccelerated) becomes $\tilde{O}(\nnz(\mat{A}) + ({\sr(\mat{A})}/{\gap^2})(s + \sqrt{\sr(\mat{A})\cdot s}))$.
  To understand the relative strength of our results, we give an example of one parameter regime where our running times are strictly better
  than the previous running times.
  When the rows are numerically sparse i.e. $s = O(1)$ although $\nnz(a_i) = d$, then our unaccelerated running time $\tilde{O}(\nnz(\mat{A}) + ({\sr(\mat{A})}/{\gap^2})\sqrt{\sr(\mat{A})})$
  gives significant improvement over the previous best unaccelerated running time of $\tilde{O}(\nnz(\mat{A}) + d(\sr(\mat{A})/{\gap^2}))$ since $\sr(\mat{A}) \leq d$.
Also, our accelerated running time of $\tilde{O}\left(nd + \frac{(nd)^{{3}/{4}}}{\sqrt{\gap}}{\sr^{3/8}}\right)$ gives significant improvement
over the previous best accelerated running time of $\tilde{O}\left(nd + \frac{(nd)^{{3}/{4}}(d \cdot\sr )^{{1}/{4}}}{\sqrt{\gap}}\right)$ since $\sr(\mat{A}) \leq d$.

Using the same techniques as in \cite{garber2016faster}, we can also achieve a gap free unaccelerated running time of $$\tilde{O}\left(\nnz(\mat{A}) + \frac{1}{\epsilon^2}\sum_i \frac{\norm{a_i}_2^2}{\lambda_1}\left(\sqrt{s_i} + \sqrt{\sr(\mat{A})}\right)\sqrt{s_i}\right)$$
and a gap free accelerated running time of $$\tilde{O}\left(\nnz(\mat{A}) + \frac{\nnz(\mat{A})^{\frac{3}{4}}}{\sqrt{\epsilon}}\left({\sum_i \frac{\norm{a_i}_2^2}{\lambda_1}\left(\sqrt{s_i} + \sqrt{\sr(\mat{A})}\right)\sqrt{s_i})}\right)^{\frac{1}{4}}\right)$$
for the regime when $\epsilon \geq \gap$.

\subsubsection{Regression}
In particular, for regression, we give an unaccelerated running time of
$$
\tilde{O}\left(\nnz(\mat{A}) + \sqrt{\kappa}\sum_i\sqrt{s_i}\frac{\norm{a_i}^2_2}{  \mu}\right)
$$
and an accelerated running time of
$$
\tilde{O}\left(\nnz(\mat{A}) + \nnz(\mat{A})^{\frac{2}{3}}{\kappa}^{\frac{1}{6}}{\left(\sum_{i \in [n]}\sqrt{s_i}\frac{\norm{a_i}^2_2}{\mu}\right)}^{\frac{1}{3}}\right)
$$
Our methods improve upon the previous best unaccelerated iterative methods of $$\tilde{O}\left(\nnz(\mat{A}) + \kappa \max_i\nnz(a_i)\right)$$ and
accelerated iterative methods of
 $$\tilde{O}\left(\nnz(\mat{A}) + d\max_i\nnz(a_i) + \sum_i\frac{\norm{a_i}_2 }{ \sqrt{\mu}} \cdot \sqrt{\sigma_i(\mat{A})}\max_i\nnz(a_i)\right)$$ where $\sigma_i(\mat{A}) = \norm{a_i}^2_{(\mat{A}^\top\mat{A})^{-1}}$.
In the simpler case of uniform row norms $\norm{a_i}_2^2 = \norm{a}_2^2$ and uniform row sparsity $s_i = s$,
  our (unaccelerated) running time becomes
 \(
 \tilde{O}(\nnz(\mat{A}) + \kappa^{{3}/{2}}\sqrt{s})
 \)

 To understand the relative strength of our results, we give an example of one parameter regime where our running times are strictly better
   than the previous running times. Consider the case where the rows are numerically sparse i.e.
$s = O(1)$ but $\max_i \nnz(a_i) = d$. In that case, our unaccelerated running time becomes
\(
\tilde{O}(\nnz(\mat{A}) + \kappa^{{3}/{2}})
\)
and the previous best unaccelerated running time becomes $\tilde{O}(\nnz(\mat{A}) + \kappa d)$.
In this case, our algorithm is stricly better then the previous one when $\kappa = o(d^{2})$.

In this case, our accelerated running time becomes $\tilde{O}\left(\nnz(\mat{A}) + \nnz(\mat{A})^{\frac{2}{3}}{\kappa}^{\frac{1}{2}}\right)$
which is stricly better than the previous best running time of $\tilde{O}\left(\nnz(\mat{A}) + d^2 + \sqrt{\kappa}d^{3/2}\right)$
for certain values of $n$.
\subsection{Overview of Our Approach}
We achieve these results by carefully modifying known techniques for finite sum optimization problem to our setting. The starting point for our algorithms is Stochastic Variance Reduced Gradient Descent (SVRG) \cite{johnson2013accelerating} a popular method for finite sum optimization. This method takes steps in the direction of negative gradient in expectation and has a convergence rate which depends on a measure of variance of the steps.

We apply SVRG to our problems where we carefully subsample the entries of the rows of the data matrix so that we can compute steps that are the gradient in expectation in time possibly sublinear in the size of the row. There is an inherent issue in such a procedure, in that this can change the shape of variance. Previous sampling methods for regression ensure that the variance can be directly related to the function error, whereas here such sampling methods give $\ell_2$ error, the bounding of which in terms of function error can be expensive.

It is unclear how to completely avoid this issue and we leave this as future work. Instead, to mitigate this issue we provide several
techniques for subsampling that ensure we can obtain significant decrease in this $\ell_2$ error for small increases in the number of samples we take per row (see Section \ref{section:sampling}). Here we crucially use that we have bounds on the numerical sparsity of rows of the data matrix and prove that we can use this to quantify this decrease.

Formally, the sampling problem we have for each row is as follows. For each row $a_i$ at any point we may receive
 some vector $x$ and need to compute a random vector $g$ with $\E[g] = a_i a_i^\top x$ and with
  $\E \norm{g}_2^2$ sufficiently bounded. Ideally, we would have that $\E \norm{g}_2^2 \leq \alpha (a_i^\top x)^2$
  for some value of $\alpha$, as previous methods do. However, instead we settle for a bound of the form
  $\E \norm{g}_2^2 \leq \alpha ({a_i}^\top x)^2 + \beta \norm{x}_2^2$.

Our sampling schemes for this problem
  works as follows: For the outer $a_i$, we sample from the coordinates with probability proportional to the coordinate's absolute value, we take a few
    (more than 1) samples to control the variance (Lemma~\ref{lemma:h1variance}). For the approximation of
    $a_i^\top x$, we always take the dot product of $x$ with large coordinates of $a_i$ and we sample from the
    rest with probability proportional to the squared value of the coordinate of $a_i$ and take more than one sample to control the variance (Lemma~\ref{lemma:h2variance}).

Carefully controlling the number of samples we take per row and picking the right distribution over rows gives our bounds for regression.
For eigenvector computation, the same broad techniques work but a little more care needs to be taken to keep
the iteration costs down due to the structure of $f_i(x) \defeq \lambda \norm{x - x_0}_2^2 - (m/2) (a_i^\top (x - x_0))^2 + b_i^\top x$.
Interestingly, for eigenvector computation the penalty from $\ell_2$ error is in some sense smaller due to the structure of the objective.
\subsection{Previous Results}
Here we briefly cover previous work on regression and eigenvector computation (Section~\ref{sec:prev_algs}), sparse finite sum optimization (Section~\ref{sec:prev_sparsity}), and matrix entrywise sparsification (Section~\ref{sec:entry_sparse}).
\subsubsection{Regression and Eigenvector Algorithms}
\label{sec:prev_algs}
There is an extensive research on regression, eigenvector computation, and finite sum optimization and far too many results to state; here we include
the algorithms with the best known running times.

The results for top eigenvector computation are stated in Table~\ref{tab:EigenVectorResults} and the results for
regression are stated in Table~\ref{tab:RegressionResults}.
The algorithms work according to the weighted sparsity measure of all rows and do not take into account
the numerical sparsity which is a natural parameter to state the running times in and is strictly better than the
$\ell_0$ sparsity.
\begin{table}[h]
\tiny
\begin{center}
\caption{Previous results for computing $\epsilon$-approximate top eigenvector (Definition \ref{problem:eigenvector}). \label{tab:EigenVectorResults}}
{\footnotesize
{\tabulinesep=0.3mm
\begin{tabular}{ |>{\tiny}P{2.2cm}|>{\tiny}P{6.0cm}|>{\tiny}P{4.5cm}| }
\hline
{\bf Algorithm} &  {\bf Runtime} & {\bf Runtime with uniform row norms and sparsity}  \\
\hline
{Power Method} & {$\tilde{O}\left(\frac{\nnz}{\gap}\right)$} & {$\tilde{O}\left(\frac{nd}{\gap}\right)$}   \\
\hline
{Lanczos Method} & {$\tilde{O}\left(\frac{\nnz}{\sqrt{\gap}}\right)$} & {$\tilde{O}\left(\frac{nd}{\sqrt{\gap}}\right)$}   \\
\hline
{Fast subspace embeddings + Lanczos method \cite{clarkson2013low}} & {$\tilde{O}\left(\nnz + \frac{d \cdot\sr }{\max{\{\gap^{2.5}, \epsilon, \epsilon^{2.5}}\}}\right)$} & {$\tilde{O}\left(nd + \frac{d\cdot \sr }{\max{\{\gap^{2.5}, \epsilon, \epsilon^{2.5}}\}}\right)$}  \\
\hline
{SVRG (assuming bounded row norms and warm start) \cite{shamir2015stochastic}} & {$\tilde{O}\left(\nnz + \frac{d \cdot\sr^2}{\gap^2}\right)$}  & {$\tilde{O}\left(nd + \frac{d \cdot\sr ^2}{\gap^2}\right)$}  \\
\hline
{Shift \& Invert Power method with SVRG \cite{garber2016faster}} & {$\tilde{O}\left(\nnz  + \frac{d\cdot \sr }{\gap^2}\right)$} & {$\tilde{O}\left(nd + \frac{d \cdot\sr }{\gap^2}\right)$}  \\
\hline
{Shift \& Invert Power method with Accelerated SVRG \cite{garber2016faster}} & {$\tilde{O}\left(\nnz + \frac{\nnz ^{\frac{3}{4}}(d\cdot \sr )^{\frac{1}{4}}}{\sqrt{\gap}}\right)$} & {$\tilde{O}\left(nd + \frac{(nd)^{\frac{3}{4}}(d \cdot\sr )^{\frac{1}{4}}}{\sqrt{\gap}}\right)$} \\
\hline
{\textbf{This paper}} & {$\tilde{O}\left(\nnz + \frac{1}{\gap^2\lambda_1}\sum_i \norm{a_i}_2^2\left(\sqrt{s_i} + \sqrt{\sr }\right)\sqrt{s_i}\right)$} & {$\tilde{O}\left(nd + \frac{sr }{\gap^2}(\sqrt{s} + \sqrt{\sr })\sqrt{s}\right)$}\\
\hline
{\textbf{This paper}} & {$\tilde{O}\left(\nnz + \frac{\nnz^{\frac{3}{4}}}{\sqrt{\gap}}({\sum_i \frac{\norm{a_i}_2^2}{\lambda_1}(\sqrt{s_i} + \sqrt{\sr })\sqrt{s_i})})^{\frac{1}{4}}\right)$
} & {$\tilde{O}\left(nd + \frac{(nd)^{\frac{3}{4}}}{\sqrt{\gap}}{\sr^{1/4}(s + \sqrt{\sr \cdot s})}^{\frac{1}{4}}\right)$
}\\
\hline
\end{tabular}
}
}
\end{center}
\end{table}
\begin{table}[h]
\tiny
\begin{center}
\caption{Previous results for solving approximate regression (Definition \ref{problem:regression}). \label{tab:RegressionResults}}
{\footnotesize
{\tabulinesep=0.3mm
\begin{tabular}{ |>{\small}P{3.1cm}|>{\small}P{6.5cm}|>{\small}P{4cm}| }
\hline
{\bf Algorithm} &  {\bf Runtime} & {\bf Runtime with uniform row norms and sparsity}  \\
\hline
{Gradient Descent} & {$\tilde{O}( \nnz  \kappa_{\max} )$} & {$\tilde{O}(nd \kappa_{\max} )$}   \\
\hline
{Conjugate Gradient Descent} & {$\tilde{O}(\nnz  \sqrt{\kappa_{\max}})$} & {$\tilde{O}(nd \sqrt{\kappa_{\max}})$}   \\
\hline
{SVRG \cite{johnson2013accelerating}} & {$\tilde{O}(\nnz  + \kappa d)$}  & {$\tilde{O}(nd + \kappa d)$}  \\
\hline
{Accelerated SVRG \cite{allen2017katyusha, frostig2015regularizing, lin2015universal}} & {$\tilde{O}(\nnz  + \sqrt{n\kappa} d)$} & {$\tilde{O}(nd + \sqrt{n\kappa} d)$}  \\
\hline
{Accelerated SVRG with leverage score sampling \cite{agarwal2017leverage}} & {$\tilde{O}( \nnz  + d\max_i\nnz(a_i) + \sum_i\frac{\norm{a_i}_2}{ \sqrt{\mu}} \cdot \sqrt{\sigma_i(\mat{A})}\max_i\nnz(a_i))$} & {$\tilde{O}(nd + d^2 + {\sqrt{\kappa}} \cdot d^{3/2})$} \\
\hline
{\textbf{This paper}} & {$\tilde{O}(\nnz  + \sqrt{\kappa}\sum_i \frac{\norm{a_i}^2_2}{\mu} \sqrt{s_i})$} & {$\tilde{O}(nd + \sqrt{\kappa^{3}} \sqrt{s})$}\\
\hline
{\textbf{This paper}} & {$\tilde{O}(\nnz + \nnz^{2/3}{\kappa}^{1/6}{(\sum_{i \in [n]} \frac{\norm{a_i}^2_2}{\mu}\sqrt{s_i})}^{1/3})$} & {$\tilde{O}(nd + (nd)^{2/3}{\kappa}^{1/2}{ s}^{1/6})$}\\
\hline
\end{tabular}
}
}
\end{center}
\end{table}
\subsubsection{Sparsity Structure} \label{sec:prev_sparsity}
There has been some prior work on attempting to improve for sparsity structure. Particularly relevant is the work of \cite{konevcny2017semi} on combining coordinate descent and sampling schemes. This paper picks unbiased estimates of the gradient at each step by first picking a function and then picking a random coordinate whose variance decreases as time increases. Unfortunately, for regression and eigenvector computation computing a partial derivative is as expensive as computing the gradient
and hence, this method does not give improved running times for these problems of regression and top eigenvector computation.
\footnote{During the final editing of this paper, we found that there was an independent result \cite{tan2018stochastic} which also investigates the power of iterative optimization algorithms with $o(d)$ iteration costs and we leave it to further research to compare these results.}

\subsubsection{Entrywise Sparsification}\label{sec:entry_sparse}
Another natural approach to yielding the results of this paper would be to simply subsample the entries of $\mat{A}$ beforehand and use this as a preconditioner to solve the problem. There have been multiple works on such entrywise sparsification and in Table~\ref{tab:Results} we provide them. If we optimistically compare them to our approach, by supposing that their sparsity bounds are uniform (i.e. every row has the same sparsity) and bound its quality as a preconditioner the best of these would give bounds of $\tilde{O}(\nnz(\mat{A}) + {\lambda_{\max}\norm{\mat{A}}_F^4}/{\lambda^3_{\min}})$ \cite{kundu2014note} and $ \tilde{O}(\nnz(\mat{A}) + {\sqrt{\lambda_{\max}}\norm{\mat{A}}_F^2\sum_{i}\sqrt{s_i}\norm{a_i}_2}/{\sqrt{n}\lambda^2_{\min}})$ \cite{arora2006fast} and $\tilde{O}(\nnz(\mat{A}) + {\norm{\mat{A}}_F^2\sum_is_i\norm{a_i}_2^2\lambda_{\max}}/{n\lambda^3_{\min}})$ \cite{achlioptas2013near} for regression. Bound obtained by \cite{kundu2014note} depends on the
the condition number square and does not depend on the numerical sparsity structure of the matrix. Bound obtained by \cite{arora2006fast} is worse as compared to our bound when compared
with matrices having equal row norms and uniform sparsity. Our running time for regression is  $\tilde{O}(\nnz(\mat{A}) + \sqrt{\kappa}\sum_i ({\norm{a_i}^2_2}/{\mu})\sqrt{s_i})$.
Our results are not always comparable to that by \cite{achlioptas2013near}. Assuming uniform sparsity and row norms, we
get that the ratio of our running time to the running time of \cite{achlioptas2013near} is $ {(\lambda_{\min}n)}/{(\sqrt{s}\lambda_{\max}\sqrt{\kappa})}$.
Depending on the values of the particular parameters, the ratio can be both greater or less than 1 and hence, the results are incomparable. Our results are
always better than that obtained by \cite{kundu2014note}.
%

% !TEX root = main.tex
\section{Sampling Techniques}\label{section:sampling}
In this section, we first give the notation that would be used throughout the paper and then state the sampling techniques used for sampling from a matrix.

\subsection{Notation}
\textbf{Vector Properties: }For $a \in \R^d$, let $s(a) = \norm{a}_1^2 / \norm{a}_2^2$ denote the numerical sparsity.
For $c \in \{1,2, \ldots, d\}$, let $\topind{a}{c}$ denote the set of indices with the  $c$ largest coordinates of $a$ in absolute value
and $\botind{a}{c} = [d]\setminus \topind{a}{c}$ i.e. everything except the indices of top $c$ co-ordinates in absolute value. Let $({\upperthres{a}{c}})_i = a_i$ if $i \in \topind{a}{c}$ and $0$ otherwise and $\lowerthres{a}{c} = a - \upperthres{a}{c}$.
 Let $\hat{e}_j$ denote the ith basis vector i.e. $(\hat{e}_j)_i = 1$ if $i = j$ and 0 otherwise.\\[0.05in]
\textbf{Other: } We let $[d]$ denote the set $\{1,2, \ldots, d\}$. We use $\tilde{O}$ notation to hide
polylogarithmic factors in the input parameters and error rates.
Refer to Section~\ref{section:problems} for other definitions.

\subsection{Sampling Lemmas}
In this section, we provide our key tools for sampling from a matrix for both regression and eigenvector computation. First, we provide a technical lemma
on numerical sparsity that we will use throughout our analysis. Then, we provide and analyze the sampling distribution we use to sample from our matrix for SVRG.
We use the same distribution for both the applications, regression and eigenvector computation and provide some of the analysis of properties of this distribution.
All proofs in this section are differed to Appendix~\ref{sec:sample_proofs}.

We begin with a lemma at the core of the proofs of our sampling techniques. The lemma essentially states that for a numerically sparse vector, most of the $\ell_2$-mass of the vector is concentrated in its top few coordinates. Consequently, if a vector is numerically sparse then we can remove a few big coordinates from it and reduce its $\ell_2$ norm considerably. Later, in our sampling schemes, we will use this lemma to bound the variance of sampling a vector.
\begin{restatable}[Numerical Sparsity]{lemma}{sparsityratio}
\label{lemma:sparsityratio}
For $a \in \R^d$ and $c \in [d]$,  we have
$\norm{\lowerthres{a}{c}}_2^2 \leq  \norm{a}_2^2 \frac{s(a)}{c}$.
\end{restatable}
The following lemmas state the sampling distribution that we use for sampling the gradient function in SVRG. Basically, since we want to approximate the gradient of $f(x) = \frac{1}{2}x \mat{A}^\top  \mat{A} x - b^\top x$ i.e. $\mat{A}^\top \mat{A} x - b$, we would like to sample $\mat{A}^\top  \mat{A} x = \sum_{i \in [n]} a_i a_i^\top x$.

We show how to perform this sampling and analyze it in several steps.
In Lemma~\ref{lemma:h1variance} we show how to sample from $a$ and then in
Lemma~\ref{lemma:h2variance} we show how to sample from $a^\top  x$.
In Lemma~\ref{lemma:h3variance} we put these together to sample from $a a^\top  x$
and in Lemma~\ref{lemma:h4variance} we put it all together to sample from
 $\mat{A}^\top  \mat{A}$. The distributions and our guarantees on them are stated below.\\
\newcommand{\Samplevec}{\texttt{Samplevec}}
\newcommand{\Samplemat}{\texttt{Samplemat}}
\newcommand{\Samplevecdotproduct}{\texttt{Samplerankonemat}}
\newcommand{\Samplerankonemat}{\texttt{Samplerankonemat}}
\newcommand{\Sampledotproduct}{\texttt{Sampledotproduct}}
\begin{minipage}[t]{8.2cm}
\DecMargin{1.5em}
 \begin{algorithm}[H]
   \begin{algorithmic}[1]\itemindent=0pc
   \caption{$\Samplevec(a, c)$\label{alg1}}
   \captionsetup{justification=centering}
   \FOR{ $t = 1 \dots c$ (i.i.d. trials)}

   \STATE randomly sample indices $j_t$ with\\
   \STATE $\Pr(j_t = j) = p_j = \frac{|a_{j}|}{\norm{a}_1}\quad \forall j \in [d]$\\
   \ENDFOR
   \STATE \textbf{Output} $\frac{1}{c}\sum_{t=1}^{c}\frac{a_{j_t}}{p_{j_t}}\hat{e}_{j_t}$ \\
   \vspace{0.15cm}
   \end{algorithmic}
  \end{algorithm}
  \IncMargin{1.5em}
\end{minipage}%
\begin{minipage}[t]{8.2cm}
%\null
\DecMargin{1.5em}
 \begin{algorithm}[H]

   \begin{algorithmic}[1]
   \caption{$\Sampledotproduct(a, x, c)$\label{alg2}}
   \FOR{ $t = 1 \dots c$ (i.i.d. trials)}
  \STATE randomly sample indices $j_t$ with\\
   \STATE $\Pr(j_t = j) = p_j = \frac{a_{j}^2}{\norm{\lowerthres{a}{c}}_2^2}\quad \forall j \in \botind{a}{c}$\\

   \ENDFOR
   \STATE \textbf{Output} $\upperthres{a}{c}^\top x + \frac{1}{c}\sum_{t=1}^{c}\frac{a_{j_t}x_{j_t}}{p_{j_t}}$
   \vspace{0.006cm}
   \end{algorithmic}
  \end{algorithm}
\IncMargin{1.5em}
\end{minipage}\\
\begin{minipage}[t]{8.2cm}
\null
\DecMargin{1.5em}
 \begin{algorithm}[H]
   \begin{algorithmic}[1]
\caption{$\Samplevecdotproduct(a, x, c)$\label{alg3}}
   \STATE $(\widehat{a})_{c}$ = $\Samplevec(a, c)$\\
   \STATE $(\widehat{a^\top x})_{c}$ = $\Sampledotproduct(a, x, c)$\\
   \STATE \textbf{Output} $(\widehat{a})_{c}(\widehat{a^\top x})_{c}$
   \vspace{1.783cm}
   \end{algorithmic}
  \end{algorithm}
  \IncMargin{1.5em}
\end{minipage}%
\begin{minipage}[t]{8.2cm}
\null
\DecMargin{1.5em}
 \begin{algorithm}[H]
   \begin{algorithmic}[1]
     \caption{$\Samplemat(\mat{A}, x, k)$\label{alg4}}
     \STATE $c_i = \sqrt{s_i}\cdot k \quad \forall i \in [n]$\\
     \STATE $M = \sum_i \norm{a_i}_2^2(1 + \frac{s_i}{c_i})$\\
     \STATE Select a row index $i$ with probability $p_i = \frac{\norm{a_i}_2^2}{M}(1 + \frac{s_i}{c_i})$\\
     \STATE $(\widehat{a_ia_i^\top x})_{c_i}$ = $\Samplerankonemat(a_i, x, c_i)$\\
     \STATE \textbf{Output} $\frac{1}{p_i}(\widehat{a_ia_i^\top x})_{c_i}$
   \end{algorithmic}
  \end{algorithm}
  \IncMargin{1.5em}
\end{minipage}
\begin{restatable}[Stochastic Approximation of $a$]{lemma}{honevariance}
\label{lemma:h1variance}
Let $a \in \R^d$ and $c \in \N$ and let our estimator be defined as $(\hat{a})_c = \Samplevec(a, c)$ (Algorithm \ref{alg1})
Then,
\[
\E[(\hat{a})_c] = a
~
\text{ and }
~
\E\left[\norm{(\hat{a})_c}_2^2\right] \leq \norm{a}_2^2\left(1 + \frac{s(a)}{c}\right)\]
\end{restatable}
\begin{restatable}[Stochastic Approximation of $a^\top x$]{lemma}{htwovariance}
\label{lemma:h2variance}
Let $a, x \in \R^d$ and $c \in [d]$, and let our estimator be defined as $(\widehat{a^\top x})_c = \Sampledotproduct(a, x, c)$ (Algorithm \ref{alg2})
Then,
\[
\E[(\widehat{a^\top x})_c] = a^\top x
~
\text{ and }
~
\E\left[(\widehat{a^\top x})^2_c\right] \leq (a^\top x)^2 + \frac{1}{c}\norm{\lowerthres{a}{c}}^2_2\norm{x}_2^2
\]
\end{restatable}
\begin{restatable}[Stochastic Approximation of $aa^\top x$]{lemma}{hthreevariance}
\label{lemma:h3variance}
Let $a, x \in \R^d$ and $c \in [d]$, and the estimator be defined as $(\widehat{aa^\top x})_c = \Samplevecdotproduct(a, x, c)$ (Algorithm \ref{alg3})
Then,
\[
\E[(\widehat{aa^\top x})_c] = aa^\top x
~
\text{ and }
~
\E\left[\norm{(\widehat{aa^\top x})_c}_2^2\right] \leq \norm{a}_2^2 \left(1 + \frac{s(a)}{c}\right)\left( (a^\top x)^2 + \frac{s(a)}{c^2}\norm{a}^2_2\norm{x}_2^2 \right)
\]
\end{restatable}
\begin{restatable}[Stochastic Approximation of $\mat{A}^\top \mat{A}x$ ]{lemma}{hfourvariance}
\label{lemma:h4variance}
Let $\mat{A} \in \R^{n \times d}$ with rows $a_1, a_2, \ldots, a_n$ and $x \in \R^d$ and
 let $(\widehat{\mat{A}^\top \mat{A}x})_k = \Samplemat(\mat{A}, x, k)$ (Algorithm \ref{alg4}) where $k$ is some parameter.
Then,
\[
\E\left[(\widehat{\mat{A}^\top \mat{A}x})_k\right] = \mat{A}^\top \mat{A}x
~\text{ and } ~
\E\left[\normFull{(\widehat{\mat{A}^\top \mat{A}x})_k}_2^2\right] \leq M\left(\norm{\mat{A}x}_2^2 + \frac{1}{k^2}\norm{\mat{A}}^2_F\norm{x}_2^2\right)
\]
\end{restatable}
%

% !TEX root = main.tex
\section{Applications}
Using the framework of SVRG defined in
Theorem \ref{thrm:SVRG} and the sampling techniques presented in Section~\ref{section:sampling}, we now state how we solve our problems of regression and top eigenvector computation.
\subsection{Eigenvector computation}
The classic method to estimate the top eigenvector of a matrix is to apply \emph{power method}. This method starts with an initial vector $x_0$ and repeatedly multiplies the vector by $\mat{A}^\top \mat{A}$ which eventually leads to convergence of the vector to the top eigenvector of the matrix $\mat{A}^\top \mat{A}$ if top eigenvalue of the matrix is well separated from the other eigenvalues i.e. $\gap$ is large enough. The number of iterations required for convergence is $O(\log(\frac{d}{\epsilon})/\gap)$. However, this method can be very slow when the gap is small.

If the gap is small, improved convergence rates can be achieved by applying power method to the matrix $\mat{B}^{-1}$ where $\mat{B} = \lambda \mat{I} - \mat{A}^\top \mat{A}$ \cite{musco2015randomized}. $\mat{B^{-1}}$ has the same largest eigenvector as $\mat{A}^\top \mat{A}$ and the eigenvalue gap is ${(\frac{1}{\lambda - \lambda_1} - \frac{1}{\lambda - \lambda_2})}/{\frac{1}{\lambda - \lambda_1}} = \frac{1}{2}$ if $\lambda \approx (1 + \gap)\lambda_1$ and thus we get a constant eigenvalue gap. Hence, if we have
a rough estimate of the largest eigenvalue of the matrix, we can get the gap parameter
as roughly constant. Section 6 of \cite{garber2016faster} shows how we can get such an estimate based on the gap free eigenvalue estimation algorithm by \cite{musco2015randomized} in running time dependent on the linear system solver of $\mat{B}$ ignoring some additional
polylogarithmic factors. However, doing power iteration on $\mat{B}^{-1}$ requires solving linear systems on $\mat{B}$ whose condition number now depends on $\frac{1}{\gap}$ and thus, solving linear system on $\mat{B}$ would become expensive now. \cite{garber2016faster} showed how we can solve the linear systems in $\mat{B}$ faster by the idea of using SVRG \cite{johnson2013accelerating} and achieve a better overall running time for top eigenvector computation. The formal theorem statement is differed to Theorem \ref{thrm:garber} in the appendix.

To achieve our improved running times for eigenvector computation, we simply use this framework for solving the eigenvector problem using
SVRG and on the top of that, give different sampling scheme (presented in Lemma~\ref{lemma:h4variance}) for SVRG for $\mat{B}^{-1}$ which reduces the runtime for
numerically sparse matrices. The following lemma states the variance bound that we get for the gradient updates for SVRG for
the top eigenvector computation problem. This is used to obtain a bound on the solving of linear systems in
$\mat{B} \defeq \lambda \mat{I} - \mat{A}^\top \mat{A}$ which will be ultimately used in solving the approximate top eigenvector problem.
The proof of this appears in Apppendix \ref{sec:app_proofs}.
\begin{restatable}[Variance bound for top eigenvector computation]{lemma}{vareigvecot} \label{lemma:vareigvecot}
Let $\nabla g(x) \defeq \lambda x - (\widehat{\mat{A}^\top \mat{A}x})_k$ be the estimator of $\nabla f(x)$
where $f(x) = \frac{1}{2}x^\top \mat{B}x - b^\top x$ and $\mat{B} = \lambda\mat{I} - \mat{A}^\top\mat{A}$ and
$(\widehat{\mat{A}^\top \mat{A}x})_k$ is the estimator of $\mat{A}^\top \mat{A}x$ defined in Lemma \ref{lemma:h4variance},
and $k = \sqrt{\sr(\mat{A})}$, then we get
\[
\E[\nabla g(x)] = (\lambda \mat{I}- \mat{A}^\top \mat{A})x
~\text{ and } ~
\E\left[\norm{\nabla g(x) - \nabla g(x^*)}_2^2\right] \leq (f(x) - f(x^*))\frac{8M}{\gap}
~\text{ and } ~
\]
\[
T = \sum_i \frac{\norm{a_i}_2^2\left(s_i + \sqrt{s_i\sr(\mat{A})}\right)}{M}
\]
where T is the average time taken in calculating $\nabla g(x)$ and
 $M = \sum_i\norm{a_i}_2^2\left(1 + \sqrt{\frac{s_i}{\sr(\mat{A})}}\right)$ assuming $ \left(1 + \frac{\gap}{150}\right)\lambda_1  \leq \lambda \leq 2\lambda_1$.
\end{restatable}
Now, using the variance of the gradient estimators and per iteration running time $T$ obtained in Lemma \ref{lemma:vareigvecot} along with the framework of SVRG \cite{johnson2013accelerating} (defined in Theorem \ref{thrm:SVRG}), we can get
constant multiplicative decrease in the error in solving linear systems in
$\mat{B} = \lambda \mat{I} - \mat{A}^\top \mat{A}$ in total running time\\
$O\left(\nnz(\mat{A}) + \frac{2}{\gap^2\lambda_1(\mat{A}^\top \mat{A})}\sum_i \norm{a_i}_2^2(\sqrt{s_i} + \sqrt{\sr(\mat{A})})\sqrt{s_i}\right)$ assuming $\lambda$ is a crude approximation to the top eigenvalue of $\mat{A}$.
The formal theorem statement (Theorem \ref{thrm:eigsvrgruntime}) and proof are differed to the appendix.
Now, using the linear system solver described above along with the shift and invert algorithmic framework, we get the following running time for top eigenvector computation problem. The proof of this theorem appears in Appendix \ref{sec:app_proofs}.
\begin{restatable}[Numerically Sparse Top Eigenvector Computation Runtime]{theorem}{eigruntime}
  \label{thrm:eigruntime} Linear system solver from Theorem \ref{thrm:eigsvrgruntime} combined with the
  shift and invert framework from \cite{garber2016faster} stated in Theorem \ref{thrm:garber} gives an algorithm which computes $\epsilon$-approximate top eigenvector (Definition \ref{problem:eigenvector})
  in total running time $$
  O\left(\left(\nnz(\mat{A}) + \frac{1}{\gap^2\lambda_1}\sum_i \norm{a_i}_2^2\left(\sqrt{s_i} + \sqrt{\sr(\mat{A})}\right)\sqrt{s_i}\right)\cdot\left( \log^2\left(\frac{d}{\gap}\right) + \log\left(\frac{1}{\epsilon}\right) \right)\right) ~
  $$
\end{restatable}
The following accelerated running time for top eigenvector computation follows from using the general acceleration framework of
\cite{frostig2015regularizing} mentioned in Theorem~\ref{thrm:acceleration} in the appendix along with the linear system solver runtime
and the proof appears in Appendix \ref{sec:app_proofs}.

\begin{restatable}[Numerically Sparse Accelerated Top Eigenvector Computation Runtime]{theorem}{acceigruntime}
   \label{thrm:acceigruntime} Linear system solver from Theorem \ref{thrm:eigsvrgruntime} combined with acceleration framework from
   \cite{frostig2015regularizing} mentioned in Theorem \ref{thrm:acceleration} and
   shift and invert framework from \cite{garber2016faster} stated in Theorem \ref{thrm:garber} gives an algorithm which computes $\epsilon$-approximate top eigenvector (Definition \ref{problem:eigenvector})
in total running time
$$
\tilde{O}\left(\nnz(\mat{A}) + \frac{\nnz(\mat{A})^{\frac{3}{4}}}{\sqrt{\gap}}\left({\sum_i \frac{\norm{a_i}_2^2}{\lambda_1}\left(\sqrt{s_i} + \sqrt{\sr(\mat{A})}\right)\sqrt{s_i})}\right)^{\frac{1}{4}}\right)
$$ where $\tilde{O}$ hides a factor of $\left( \log^2\left(\frac{d}{\gap}\right) + \log\left(\frac{1}{\epsilon}\right)\right)\log\left(\frac{d }{\gap}\right)$.
\end{restatable}

\subsection{Linear Regression}
In linear regression, we want to minimize $\frac{1}{2}\norm{\mat{A}x - b}_2^2$ which is equivalent to minimizing
$$
\frac{1}{2}x^\top \mat{A}^\top \mat{A}x - x^\top \mat{A}^\top b = \frac{1}{2}\sum_i x^\top a_i{a_i}^\top x - x^\top \mat{A}^\top b
$$
and hence, we can apply the framework of SVRG \cite{johnson2013accelerating} (stated in Theorem \ref{thrm:SVRG}) for solving it. However, instead of selecting a complete row for calculating the gradient, we only select a
few entries from the row to achieve lower cost per iteration. In particular, we
use the distribution defined in Lemma \ref{lemma:h4variance}. Note that the sampling probabilities depend on $\lambda_d$ and we need to know a
constant factor approximation of $\lambda_d$ for the scheme to work. For most of the ridge regression problems, we know a lower bound on the
value of $\lambda_d$ and we can get an approximation by doing a binary search over all the values and paying an extra logarithmic factor.
The following lemma states the sampling distribution which we use for approximating the true gradient and the corresponding variance that we obtain.
The proof of this appears in Appendix \ref{sec:app_proofs}.
\begin{restatable}[Variance Bound for Regression]{lemma}{varlinreg} Let $\nabla g(x) \defeq (\widehat{\mat{A}^\top \mat{A}x})_k$ be the estimator of $\nabla f(x)$ where $f(x) = \frac{1}{2}\norm{\mat{A}x - b}_2^2$ and $(\widehat{\mat{A}^\top \mat{A}x})_k$ is the estimator for $\mat{A}^\top \mat{A}x$ defined in Lemma \ref{lemma:h4variance} and $k = \sqrt{\kappa} $, assuming $\kappa \leq d^2$ we get
\label{lemma:varlinreg}
\[
\E[\nabla g(x)] = \mat{A}^\top \mat{A}x
~\text{ and } ~
\E\left[\normFull{\nabla g(x) - \nabla g(x^*)}_2^2\right] \leq M(f(x) - f(x^*))
~\text{ and } ~
T = \frac{\sqrt{\kappa}}{M}\sum_{i \in [n]} \norm{a_i}_2^2\sqrt{s_i}
\]
where T is the average time taken in calculating $\nabla g(x)$ where $M = \sum_i\norm{a_i}_2^2\left(1 + \sqrt{\frac{s_i}{\kappa}}\right)$.
\end{restatable}
Using the variance bound obtained in Lemma \ref{lemma:varlinreg} and the framework of SVRG stated
in Theorem \ref{thrm:SVRG} for solving approximate linear systems, we show how we can obtain an
algorithm for solving approximate regression in time which is faster in certain regimes when
the corresponding matrix is numerically sparse. The proof of this appears in Appendix \ref{sec:app_proofs}.
\begin{restatable}[Numerically Sparse Regression Runtime]{theorem}{runtimeregress}
\label{thm:runtime_regress}
For solving $\epsilon$-approximate regression (Definition \ref{problem:regression}), if $\kappa \leq d^2$, SVRG framework from Theorem \ref{thrm:SVRG} and
the variance bound from Lemma \ref{lemma:varlinreg} gives an algorithm with running time $O\left(\left(\nnz(\mat{A}) + \sqrt{\kappa}\sum_{i \in [n]} \frac{\norm{a_i}^2_2}{\mu}\sqrt{s_i}\right)\log\left(\frac{1}{\epsilon}\right)\right)$.
\end{restatable}
Combined with the additional acceleration framework mentione in Theorem \ref{thrm:acceleration}, we can get an accelerated
algorithm for solving regression. The proof of this appears in Appendix \ref{proof:acc_runtime_regress}.
\begin{restatable}[Numerically Sparse Accelerated Regression Runtime]{theorem}{accruntimeregress}
\label{thm:acc_runtime_regress}
  For solving $\epsilon$-approximate regression (Definition \ref{problem:regression}) if $\kappa \leq d^2$, SVRG framework from Theorem \ref{thrm:SVRG},
   acceleration framework from Theorem \ref{thrm:acceleration} and
  the variance bound from Lemma \ref{lemma:varlinreg} gives an algorithm with running time
$$
  O\left(\left(\nnz(\mat{A}) + \nnz(\mat{A})^{\frac{2}{3}}{\kappa}^{\frac{1}{6}}{\left(\sum_{i \in [n]} \frac{\norm{a_i}^2_2}{\mu}\sqrt{s_i}\right)}^{\frac{1}{3}}\log({\kappa})\right)\log\left(\frac{1}{\epsilon}\right)\right) ~.
  $$
\end{restatable}

% !TEX root = main.tex
\section*{Acknowledgments}
We would like to thank the anonymous reviewers who helped improve the readability and presentation of this draft by providing
many helpful comments.

\bibliographystyle{plain}
\bibliography{refs}

\begin{thebibliography}{10}

\bibitem{achlioptas2013near}
Dimitris Achlioptas, Zohar~S Karnin, and Edo Liberty.
\newblock Near-optimal entrywise sampling for data matrices.
\newblock In {\em Advances in Neural Information Processing Systems}, pages
  1565--1573, 2013.

\bibitem{achlioptas2007fast}
Dimitris Achlioptas and Frank McSherry.
\newblock Fast computation of low-rank matrix approximations.
\newblock {\em Journal of the ACM (JACM)}, 54(2):9, 2007.

\bibitem{agarwal2017leverage}
Naman Agarwal, Sham Kakade, Rahul Kidambi, Yin~Tat Lee, Praneeth Netrapalli,
  and Aaron Sidford.
\newblock Leverage score sampling for faster accelerated regression and erm.
\newblock {\em arXiv preprint arXiv:1711.08426}, 2017.

\bibitem{allen2017katyusha}
Zeyuan Allen-Zhu.
\newblock Katyusha: The first direct acceleration of stochastic gradient
  methods.
\newblock In {\em Proceedings of the 49th Annual ACM SIGACT Symposium on Theory
  of Computing}, pages 1200--1205. ACM, 2017.

\bibitem{arora2006fast}
Sanjeev Arora, Elad Hazan, and Satyen Kale.
\newblock A fast random sampling algorithm for sparsifying matrices.
\newblock In {\em Approximation, Randomization, and Combinatorial Optimization.
  Algorithms and Techniques}, pages 272--279. Springer, 2006.

\bibitem{bottou2004large}
L{\'e}on Bottou and Yann~L Cun.
\newblock Large scale online learning.
\newblock In {\em Advances in neural information processing systems}, pages
  217--224, 2004.

\bibitem{clarkson2013low}
Kenneth~L Clarkson and David~P Woodruff.
\newblock Low rank approximation and regression in input sparsity time.
\newblock In {\em Proceedings of the forty-fifth annual ACM symposium on Theory
  of computing}, pages 81--90. ACM, 2013.

\bibitem{drineas2011note}
Petros Drineas and Anastasios Zouzias.
\newblock A note on element-wise matrix sparsification via a matrix-valued
  bernstein inequality.
\newblock {\em Information Processing Letters}, 111(8):385--389, 2011.

\bibitem{frostig2015regularizing}
Roy Frostig, Rong Ge, Sham Kakade, and Aaron Sidford.
\newblock Un-regularizing: approximate proximal point and faster stochastic
  algorithms for empirical risk minimization.
\newblock In {\em International Conference on Machine Learning}, pages
  2540--2548, 2015.

\bibitem{garber2016faster}
Dan Garber, Elad Hazan, Chi Jin, Cameron Musco, Praneeth Netrapalli, Aaron
  Sidford, et~al.
\newblock Faster eigenvector computation via shift-and-invert preconditioning.
\newblock In {\em International Conference on Machine Learning}, pages
  2626--2634, 2016.

\bibitem{gittens2009error}
Alex Gittens and Joel~A Tropp.
\newblock Error bounds for random matrix approximation schemes.
\newblock {\em arXiv preprint arXiv:0911.4108}, 2009.

\bibitem{johnson2013accelerating}
Rie Johnson and Tong Zhang.
\newblock Accelerating stochastic gradient descent using predictive variance
  reduction.
\newblock In {\em Advances in neural information processing systems}, pages
  315--323, 2013.

\bibitem{konevcny2017semi}
Jakub Kone{\v{c}}n{\`y}, Zheng Qu, and Peter Richt{\'a}rik.
\newblock Semi-stochastic coordinate descent.
\newblock {\em Optimization Methods and Software}, 32(5):993--1005, 2017.

\bibitem{kronmal1979alias}
Richard~A Kronmal and Arthur~V Peterson.
\newblock The alias and alias-rejection-mixture methods for generating random
  variables from probability distributions.
\newblock In {\em Proceedings of the 11th conference on Winter
  simulation-Volume 1}, pages 269--280. IEEE Press, 1979.

\bibitem{kundu2014note}
Abhisek Kundu and Petros Drineas.
\newblock A note on randomized element-wise matrix sparsification.
\newblock {\em arXiv preprint arXiv:1404.0320}, 2014.

\bibitem{lin2015universal}
Hongzhou Lin, Julien Mairal, and Zaid Harchaoui.
\newblock A universal catalyst for first-order optimization.
\newblock In {\em Advances in Neural Information Processing Systems}, pages
  3384--3392, 2015.

\bibitem{musco2015randomized}
Cameron Musco and Christopher Musco.
\newblock Randomized block krylov methods for stronger and faster approximate
  singular value decomposition.
\newblock In {\em Advances in Neural Information Processing Systems}, pages
  1396--1404, 2015.

\bibitem{musco2017spectrum}
Cameron Musco, Praneeth Netrapalli, Aaron Sidford, Shashanka Ubaru, and David~P
  Woodruff.
\newblock Spectrum approximation beyond fast matrix multiplication: Algorithms
  and hardness.
\newblock {\em arXiv preprint arXiv:1704.04163}, 2017.

\bibitem{nguyen2010tensor}
Nam~H Nguyen, Petros Drineas, and Trac~D Tran.
\newblock Tensor sparsification via a bound on the spectral norm of random
  tensors.
\newblock {\em arXiv preprint arXiv:1005.4732}, 2010.

\bibitem{nguyen2009matrix}
NH~Nguyen, Petros Drineas, and TD~Tran.
\newblock Matrix sparsification via the khintchine inequality.
\newblock 2009.

\bibitem{shalev2013stochastic}
Shai Shalev-Shwartz and Tong Zhang.
\newblock Stochastic dual coordinate ascent methods for regularized loss
  minimization.
\newblock {\em Journal of Machine Learning Research}, 14(Feb):567--599, 2013.

\bibitem{shamir2015stochastic}
Ohad Shamir.
\newblock A stochastic pca and svd algorithm with an exponential convergence
  rate.
\newblock In {\em International Conference on Machine Learning}, pages
  144--152, 2015.

\bibitem{tan2018stochastic}
Conghui Tan, Tong Zhang, Shiqian Ma, and Ji~Liu.
\newblock Stochastic primal-dual method for empirical risk minimization with
  $\mathcal{O}(1)$ per-iteration complexity.
\newblock {\em arXiv preprint arXiv:1811.01182}, 2018.

\end{thebibliography}

\appendix

% !TEX root = main.tex
% !TEX root = main.tex
\section{Preliminaries}
\subsection{SVRG}

We state the SVRG framework which was originally given by \cite{johnson2013accelerating}. Our algorithms for solving regression and top eigenvector computation both involve using this framework to solve linear systems. SVRG is used to get linear convergence for stochastic gradient descent by taking gradient updates which equals the exact gradient in expectation but which have reduced variance which goes down to 0 as we reach near the optimum. This specific statement of the results is a slight modification of the statement presented in \cite{garber2016faster}. The proof is analogous to the one from \cite{garber2016faster}. We are stating it here for completeness.

\begin{theorem}[SVRG for Sums of Non-Convex functions] \label{thrm:SVRG} Let D be a distribution over functions, $g_1, g_2, \ldots \in \R^d \rightarrow \R^d$. Let $\nabla f(x) - \nabla f(y) = \E_{g_{i_k} \sim D} \nabla g_{i_k}(x) - \nabla g_{i_k}(y) \quad \forall x,y \in \R^d$ and let $x^* = \argmin_{x \in \R^d} f(x)$. Suppose that starting from some initial point $x_0 \in \R^d$ in each iteration k, we let
$$x_{k+1} :=  x_k - \eta (\nabla g_{i_k}(x_k) - \nabla g_{i_k}(x_0)) + \eta \nabla f(x_0)$$ where $g_{i_k} \sim D$ independently at random for some $\eta$.

If $f$ is $\mu$-strongly convex and if for all $x \in \R^d$, we have
\begin{align}
\E_{g_{i_k} \sim D} \norm{\nabla g_{i_k}(x) - \nabla g_{i_k}(x^*)}_2^2 \leq 2\sigma^2(f(x) - f(x^*)) \label{eqn:svrgclaim}
\end{align}
where $\sigma^2$ we call the \emph{variance parameter}, then for all $m \geq 1$, we have
\[
\E\left[\frac{1}{m} \sum_{k \in [m]} f(x_k) - f(x^*) \right] \leq \frac{1}{1 - 2\eta\sigma^2}\left(\frac{1}{m\eta\mu} + 2\eta\sigma^2\right)(f(x_0) - f(x^*))
\]
Consequently, if we pick $\eta$ to be a sufficiently small multiple of $\frac{1}{\sigma^2}$ then when $m = O(\frac{\sigma^2}{\mu})$, we can decrease the error by a constant multiplicative factor in expectation.
\end{theorem}

\begin{proof}
Using the fact that we have, $\nabla f(x) - \nabla f(y) = \E_{g_{i_k} \sim D} \nabla g_{i_k}(x) - \nabla g_{i_k}(y) \quad \forall x,y \in \R^d$, we have that:
\begin{align}
\E_{g_{i_k} \sim D}\norm{x_{k+1} - x^*}_2^2 &=  \E_{g_{i_k} \sim D}\norm{x_{k} - \eta (\nabla g_{i_k}(x_k) - \nabla g_{i_k}(x_0) + \nabla f(x_0)) - x^*}_2^2 \nonumber\\
&= \E_{g_{i_k} \sim D}\norm{x_{k} - x^*}_2^2 - 2\eta \E_{g_{i_k} \sim D}(\nabla g_{i_k}(x_k) - \nabla g_{i_k}(x_0) + \nabla f(x_0))^\top (x_k - x^*) + \nonumber\\
& \eta^2 \E_{g_{i_k} \sim D}\norm{\nabla g_{i_k}(x_k) - \nabla g_{i_k}(x_0) + \nabla f(x_0)}_2^2\nonumber\\
&= \E_{g_{i_k} \sim D}\norm{x_{k} - x^*}_2^2 - 2\eta \nabla f(x_k)^\top (x_k - x^*)  + \nonumber\\
& \eta^2 \E_{g_{i_k} \sim D}\norm{\nabla g_{i_k}(x_k) - \nabla g_{i_k}(x_0) + \nabla f(x_0)}_2^2 \label{eqn:svrg:inter1}
\end{align}

Now, using $\norm{x + y}_2^2 \leq 2\norm{x}_2^2 + 2\norm{y}_2^2$, we get:
\begin{align}
\E_{g_{i_k} \sim D}\norm{\nabla g_{i_k}(x_k) - \nabla g_{i_k}(x_0) + \nabla f(x_0)}_2^2 &\leq 2\E_{g_{i_k} \sim D}\norm{\nabla g_{i_k}(x_k) - \nabla g_{i_k}(x^*)}_2^2 + \nonumber\\
& 2\E_{g_{i_k} \sim D}\norm{\nabla g_{i_k}(x_0) - \nabla g_{i_k}(x^*) - \nabla f(x_0)}_2^2 \label{eqn:svrg:inter}
\end{align}

Now, we know that $\nabla f(x^*) = 0$ and using $\E\norm{x - \E x}^2 \leq \E\norm{x}_2^2$
\begin{align}
\E_{g_{i_k} \sim D}\norm{\nabla g_{i_k}(x_0) - \nabla g_{i_k}(x^*) - \nabla f(x_0)}_2^2 &= \E_{g_{i_k} \sim D}\norm{\nabla g_{i_k}(x_0) - \nabla g_{i_k}(x^*) - (\nabla f(x_0) - \nabla f(x^*))}_2^2 \nonumber\\
&\leq \E_{g_{i_k} \sim D}\norm{\nabla g_{i_k}(x_0) - \nabla g_{i_k}(x^*)}_2^2 \label{eqn:var1}
\end{align}

Now, using  \eqref{eqn:svrgclaim} and \eqref{eqn:var1} in \eqref{eqn:svrg:inter}, we get:
\begin{align}
\E_{g_{i_k} \sim D}\norm{\nabla g_{i_k}(x_k) - \nabla g_{i_k}(x_0) + \nabla f(x_0)}_2^2
&\leq 4\sigma^2(f(x_k) - f(x^*)) + 4\sigma^2(f(x_0) - f(x^*)) \nonumber\\
&\leq 4\sigma^2(f(x_k) - f(x^*) + f(x_0) - f(x^*)) \label{eqn:svrg:inter2}
\end{align}

Using the convexity of $f$, we get $f(x^*) - f(x_k) \geq \nabla f(x_k)^\top (x^* - x_k)$, using this and \eqref{eqn:svrg:inter2} in \eqref{eqn:svrg:inter1}, we get
\begin{align}
\E_{g_k \sim D}\norm{x_{k+1} - x^*}_2^2 &\leq \norm{x_{k} - x^*}_2^2 - 2\eta \nabla f(x_k)^\top (x_k - x^*)
+ 4\eta^2\sigma^2(f(x_k) - f(x^*) + f(x_0) - f(x^*)) \nonumber\\
&\leq \norm{x_{k} - x^*}_2^2 - 2\eta (f(x_k) - f(x^*)) + 4\eta^2\sigma^2(f(x_k) - f(x^*) + f(x_0) - f(x^*)) \nonumber\\
&\leq \norm{x_{k} - x^*}_2^2 - 2\eta(1 - 2\eta\sigma^2) (f(x_k) - f(x^*))  + 4\eta^2\sigma^2(f(x_0) - f(x^*)) \nonumber
\end{align}

Rearranging, we get that
\begin{align}
2\eta(1 - 2\eta\sigma^2) (f(x_k) - f(x^*)) &\leq \norm{x_{k} - x^*}_2^2 - \E_{g_k \sim D}\norm{x_{k+1} - x^*}_2^2  + 4\eta^2\sigma^2(f(x_0) - f(x^*)) \nonumber
\end{align}

Summing over all iterations and taking expectations, we get
\begin{align}
2\eta(1 - 2\eta\sigma^2) \E[\sum_{k \in [m]}(f(x_k) - f(x^*))] &\leq \norm{x_{0} - x^*}_2^2 + 4m\eta^2\sigma^2(f(x_0) - f(x^*)) \nonumber
\end{align}

Now, using strong convexity, we get that $\norm{x_0 - x^*}_2^2 \leq \frac{2}{\mu}(f(x_0) - f(x^*))$ and using this we get:
\begin{align}
2\eta(1 - 2\eta\sigma^2) \E\left[\sum_{k \in [m]}(f(x_k) - f(x^*))\right] &\leq \frac{2}{\mu}(f(x_0) - f(x^*)) + 4m\eta^2\sigma^2(f(x_0) - f(x^*)) \nonumber\\
\E\left[\frac{1}{m}\sum_{k \in [m]}(f(x_k) - f(x^*))\right] &\leq \frac{1}{1 - 2\eta\sigma^2}\left(\frac{1}{m\eta\mu} + 2\eta\sigma^2\right)(f(x_0) - f(x^*)) \nonumber
\end{align}
\end{proof}

\subsection{Acceleration}

Below is a Theorem from \cite{frostig2015regularizing} which shows how can we accelerate an ERM problem where the objective is strongly convex and each of the individual components is smooth in a black box fashion by solving many regularized version of the problems. We will use this theorem to give accelerated runtimes for our problems of regression and top eigenvector computation.

\begin{theorem} [Accelerated Approximate Proximal Point, Theorem 1.1 of \cite{frostig2015regularizing}]
\label{thrm:acceleration}
Let $f: \R^n \rightarrow \R$ be a $\mu$ strongly convex function and suppose that for all $x_0 \in \R^n$, $c > 0$, $\lambda > 0$, we can compute a possibly random $x_c \in \R^n$ such that
$$\E f(x_c) - \min_{x \in \R^n}
\left\{f(x) + \frac{\lambda}{2}\norm{x - x_0}_2^2 \right\}
\leq \frac{1}{c}\left[f(x_0) - \min_{x \in \R^n} \{f(x) + \frac{\lambda}{2}\norm{x - x_0}_2^2\}\right]$$ in time $T_c$. Then, given any $x_0, c>0, \lambda \geq 2\mu$, we can compute $x_1$ such that
$$\E f(x_1) - \min_x f(x) \leq \frac{1}{c}[f(x_0) - \min_x f(x)]$$ in time $O\left(T_{4({\frac{2\lambda + \mu}{\mu}})^{\frac{3}{2}}}\sqrt{\lceil\frac{\lambda}{\mu}\rceil}\log(c)\right)$.
\end{theorem}

\section{Proofs}
\subsection{Sampling Techniques Proofs}
\label{sec:sample_proofs}
First, we provide the following Lemma~\ref{lemma:functiondiff} which will be used later in the proofs to relate the difference between function values at any point $x$ and the optimal point $x^*$ to the $\mat{A}^\top \mat{A}$ norm of difference between the two points. This is key to relating the error from sampling to function error. Note that this is standard and well known.
\begin{lemma}
\label{lemma:functiondiff}
Let $f(x) = \frac{1}{2}\norm{\mat{A}x - b}_2^2$ and $x^* = \argmin f(x)$, then $$2(f(x) - f(x^*)) = \norm{\mat{A}(x - x^*)}_2^2$$
\end{lemma}
\begin{proof}
\label{proof:functiondiff}
We know $\nabla f(x^*) = 0$ since $x^* = \argmin f(x)$,
thus, we get that $\mat{A}^\top (\mat{A}x^* - b) = 0$ or $\mat{A}^\top \mat{A}x^* = \mat{A}^\top b$.
Now,
\begin{equation*}
\begin{aligned}
2(f(x) - f(x^*)) &= \norm{\mat{A}x - b}_2^2 - \norm{\mat{A}x^* - b}_2^2\\
&= (\mat{A}x - b - \mat{A}x^* + b )^\top (\mat{A}x - 2b + \mat{A}x^*)\\
&= (x - x^*)^\top \mat{A}^\top (\mat{A}x - 2b + \mat{A}x^*)\\
&= (x - x^*)^\top (\mat{A}^\top \mat{A}x - 2\mat{A}^\top b + \mat{A}^\top \mat{A}x^*)\\
&= (x - x^*)^\top (\mat{A}^\top \mat{A}x - 2\mat{A}^\top \mat{A}x^* + \mat{A}^\top \mat{A}x^*)\\
&= (x - x^*)^\top \mat{A}^\top \mat{A}(x - x^*)\\
&= \norm{\mat{A}(x - x^*)}_2^2\\
\end{aligned}
\end{equation*}
\end{proof}

Now, we provide the proof for technical lemma
on numerical sparsity that we used throughout our analysis.
\sparsityratio*
\begin{proof}
   We can assume without loss of generality that $|a_i| \geq |a_j| $ whenever $i < j$ i.e.
  the indices are sorted in descending order of the absolute values.
\begin{align*}
\frac{\norm{\lowerthres{a}{c}}_2^2}{\norm{a}_2^2}
&= \frac{a_{c+1}^2 + a_{c+2}^2 + \cdots + a_{d}^2}{\norm{a}_2^2}
\leq \frac{|a_{c+1}|(|a_{c+1}| + |a_{c+2}| + \cdots + |a_{d}|)}{\norm{a}_2^2} \\
&\leq \frac{|a_{c+1}|\norm{a}_1}{\norm{a}_2^2}
\leq \frac{\norm{a}_1\norm{a}_1}{c\norm{a}_2^2}
\leq \frac{\norm{a}^2_1}{c\norm{a}_2^2}
\leq \frac{s(a)}{c}
\end{align*}
\end{proof}
Now, we analyze the sampling distribution we use to sample from our matrix for SVRG.
\honevariance*
\begin{proof}
\label{proof:h1variance}
Since, $(\hat{a})_c = \Samplevec(a, c)$, we can also write this as
$(\hat{a})_c = \frac{1}{c}\sum^c_{i = 1}X_i$ where $\{X_i\}$ are sampled i.i.d. such that $\Pr(X_i = \frac{a_{j}}{p_{j}}\hat{e}_j) = p_j = \frac{|a_{j}|}{\norm{a}_1}\quad \forall j \in [d]$.

Calculating first and second moments of random variable $X_i$, we get that
\begin{align}
\E[X_i] &= \sum_{j \in d} p_j\frac{a_{j}}{p_{j}}\hat{e}_j = \sum_{j \in d} \hat{e}_ja_{j}\nonumber \\
&= a \label{eqn:h21}
\end{align}
\begin{align}
\E\left[\norm{X_i}_2^2\right] &= \sum_{j \in [d]} p_j\left(\frac{a_{j}\hat{e}_j}{p_{j}}\right)^2
= \sum_{j \in [d]} \frac{a_{j}^2}{p_{j}}
= \norm{a}_1\sum_{j \in [d]} \frac{a_{j}^2}{|a_{j}|}
= \norm{a}_1\sum_{j \in [d]} |a_{j}|\nonumber\\
= \norm{a}^2_1 \label{eqn:h22}
\end{align}
Now, using the calculated moments in \eqref{eqn:h21} and \eqref{eqn:h22}, to calculate the first and second moments of $(\hat{a})_c$
\begin{equation*}
\begin{aligned}
\E[(\hat{a})_c] &= \E\left[\frac{1}{c}\sum_{i \in [c]}X_i\right]
&= \frac{1}{c}\sum_{i \in [c]}\E[X_i]
&=  \frac{1}{c}\sum_{i \in [c]}a
&= a\\
\end{aligned}
\end{equation*}
\begin{equation*}
\begin{aligned}
\E\left[\norm{(\hat{a})_c}_2^2\right] &= \E\left[\normFull{\frac{1}{c}\sum_{i \in [c]}X_i }_2^2\right]\\
&= \frac{1}{c^2}\E\left[\sum_{i \in [c]}\norm{X_i}_2^2 + \sum_{i, j \in [c], i \neq j}X_i^\top X_j\right]\\
\end{aligned}
\end{equation*}
Using the moments for random variable $X_i$ calculated in \eqref{eqn:h21} and \eqref{eqn:h22} and independence of $X_i$ and $X_j$ for $i \neq j$, we get that
\begin{equation*}
\begin{aligned}
\E\left[\norm{(\hat{a})_c}_2^2\right] &= \frac{1}{c^2}\left(\sum_{i \in [c]}\norm{a}^2_1 + \sum_{i,j \in [c], i \neq j}a^\top a\right)\\
&= \frac{1}{c^2}\left(c\norm{a}^2_1 + c(c-1)a^\top a\right)\\
&= \frac{1}{c}\left(\norm{a}^2_1 + (c-1)a^\top a\right)\\
\end{aligned}
\end{equation*}
Using $s(a) = \norm{a}_1^2 / \norm{a}_2^2$
\begin{equation*}
\begin{aligned}
\E[\norm{(\hat{a})_c}_2^2] = \norm{a}^2_2\frac{1}{c}\left(s(a) + (c-1)\right)
\leq \norm{a}^2_2\left(1 + \frac{s(a)}{c}\right)
\end{aligned}
\end{equation*}
\end{proof}
\htwovariance*
\begin{proof}
\label{proof:h2variance}
Since $(\widehat{a^\top x})_c = \Sampledotproduct(a, x, c)$, we can also write this as
$(\widehat{a^\top x})_c = \upperthres{a}{c}^\top x + \frac{1}{c}\sum^c_{i = 1}X_i(x)$ where $\{X_i\}$ are sampled i.i.d. such that for each $X_i$, $Pr(X_i = \frac{a_{k}x_k}{p_{k}}) = p_{k} = \frac{a_{k}^2}{\norm{\lowerthres{a}{c}}^2_2}\quad \forall k \in \botind{a}{c}$.

Calculating first and second moments of random variable $X_i$, we get that
\begin{align}
\E[X_i] &= \sum_{k \in \botind{a}{c}} p_k\frac{a_{k}x_k}{p_{k}} = \sum_{k \in \botind{a}{c}} a_{k}x_k\nonumber\\
&= \lowerthres{a}{c}^\top x \label{eqn:h11}
\end{align}
\begin{align}
\E\left[\norm{X_i}_2^2\right] &= \sum_{k \in \botind{a}{c}} p_k\left(\frac{a_{k}x_k}{p_{k}}\right)^2
= \sum_{k \in \botind{a}{c}} \frac{a_{k}^2x_k^2}{p_{k}}
= \sum_{k \in \botind{a}{c}} \frac{\norm{\lowerthres{a}{c}}^2_2a_{k}^2x_k^2}{a_{k}^2}
= \sum_{k \in \botind{a}{c}} \norm{\lowerthres{a}{c}}^2_2x_k^2\nonumber\\
&\leq \norm{\lowerthres{a}{c}}^2_2\norm{x}_2^2 \label{eqn:h12}
\end{align}
Using the moments calculated in \eqref{eqn:h11} and \eqref{eqn:h12}, we calculate the first and second moments of the estimator $(\widehat{a^\top x})_c$
\begin{equation*}
\begin{aligned}
\E[(\widehat{a^\top x})_c] &= \E\left[\upperthres{a}{c}^\top x + \frac{1}{c}\sum_{i \in [c]}X_i\right]\\
&= \upperthres{a}{c}^\top x + \frac{1}{c}\sum_{i \in [c]}\E[X_i]\\
&= \upperthres{a}{c}^\top x + \frac{1}{c}\sum_{i \in [c]} \lowerthres{a}{c}^\top x\\
&= \upperthres{a}{c}^\top x + \lowerthres{a}{c}^\top x\\
&= a^\top x\\
\end{aligned}
\end{equation*}
\begin{equation*}
\begin{aligned}
\E\left[(\widehat{a^\top x})^2_c\right] &= \E\left[\left(\upperthres{a}{c}^\top x + \frac{1}{c}\sum_{i \in [c]}X_i\right)^2\right] \\
&= \E\left[(\upperthres{a}{c}^\top x)^2 + 2\upperthres{a}{c}^\top x\frac{1}{c}\sum_{i \in [c]}X_i + \frac{1}{c^2}\left(\sum_{i \in [c]}X_i\right)^2\right] \\
&= \E\left[(\upperthres{a}{c}^\top x)^2\right] + \E\left[2\upperthres{a}{c}^\top x\frac{1}{c}\sum_{i \in [c]}X_i\right] + \E\left[\frac{1}{c^2}\left(\sum_{i \in [c]}X_i\right)^2\right] \\
&= (\upperthres{a}{c}^\top x)^2 + 2\upperthres{a}{c}^\top x\frac{1}{c}\sum_{i \in [c]}\E[X_i] + \E\left[\frac{1}{c^2}\left(\sum_{i \in [c]}X_i\right)^2\right] \\
\end{aligned}
\end{equation*}
Using the expectation of the random variable $X_i$, calculated in \eqref{eqn:h11}
\begin{equation*}
\begin{aligned}
\E\left[(\widehat{a^\top x})^2_c\right] &= (\upperthres{a}{c}^\top x)^2 + 2\upperthres{a}{c}^\top x\frac{1}{c}\sum_{i \in [c]}\lowerthres{a}{c}^\top x + \E\left[\frac{1}{c^2}\left(\sum_{i \in [c]}X_i\right)^2\right] \\
&= (\upperthres{a}{c}^\top x)^2 + 2\upperthres{a}{c}^\top x\lowerthres{a}{c}^\top x + \frac{1}{c^2}\E\left[\sum_{i \in [c]}X_i^2 + \sum_{i,j \in [c], i \neq j}X_i \cdot X_j\right] \\
\end{aligned}
\end{equation*}
Using the independence of $X_i$ and $X_j$, we get that
\begin{equation*}
\begin{aligned}
\E[(\widehat{a^\top x})^2_c] &= (\upperthres{a}{c}^\top x)^2 + 2\upperthres{a}{c}^\top x\lowerthres{a}{c}^\top x + \frac{1}{c^2}\left(\sum_{i \in [c]}\E\left[X_i^2\right] + \sum_{i,j \in [c], i \neq j}\E[X_i] \cdot \E[X_j]\right) \\
\end{aligned}
\end{equation*}
Using the first and second moments of the random variable $X_i$, calculated in \eqref{eqn:h11} and \eqref{eqn:h12}.
\begin{equation*}
\begin{aligned}
\E[(\widehat{a^\top x})^2_c] &= (\upperthres{a}{c}^\top x)^2 + 2\upperthres{a}{c}^\top x\lowerthres{a}{c}^\top x \nonumber\\
&+ \frac{1}{c^2}\left(\sum_{i \in [c]}\norm{\lowerthres{a}{c}}^2_2\norm{x}_2^2 + \sum_{i,j \in [c], i \neq j}\lowerthres{a}{c}^\top x \lowerthres{a}{c}^\top x\right) \\
&= (\upperthres{a}{c}^\top x)^2 + 2\upperthres{a}{c}^\top x\lowerthres{a}{c}^\top x \nonumber\\
&+ \frac{1}{c^2}\left(c\norm{\lowerthres{a}{c}}^2_2\norm{x}_2^2 + c(c-1)\lowerthres{a}{c}^\top x \lowerthres{a}{c}^\top x\right) \\
&= (\upperthres{a}{c}^\top x)^2 + 2\upperthres{a}{c}^\top x\lowerthres{a}{c}^\top x + \frac{1}{c}\norm{\lowerthres{a}{c}}^2_2\norm{x}_2^2 + \left(1 - \frac{1}{c}\right)\lowerthres{a}{c}^\top x \lowerthres{a}{c}^\top x \\
&\leq (\upperthres{a}{c}^\top x)^2 + 2\upperthres{a}{c}^\top x\lowerthres{a}{c}^\top x + \frac{1}{c}\norm{\lowerthres{a}{c}}^2_2\norm{x}_2^2 + \lowerthres{a}{c}^\top x \lowerthres{a}{c}^\top x\\
\end{aligned}
\end{equation*}
Using  $(\upperthres{a}{c}^\top x)^2 +  2\upperthres{a}{c}^\top x\lowerthres{a}{c}^\top x + \lowerthres{a}{c}^\top x \lowerthres{a}{c}^\top x = (a^\top x)^2$, we get that
\begin{equation*}
\begin{aligned}
\E[(\widehat{a^\top x})^2_c] &\leq (a^\top x)^2 + \frac{1}{c}\norm{\lowerthres{a}{c}}^2_2\norm{x}_2^2\\
\end{aligned}
\end{equation*}
\end{proof}
\hthreevariance*
\begin{proof}
\label{proof:h3variance}
Since, $(\widehat{aa^\top x})_c = (\widehat{a})_c(\widehat{a^\top x})_c$ where $(\widehat{a})_c, (\widehat{a^\top x})_c$ are the estimators for $a$ and $a^\top x$ defined in Lemma \ref{lemma:h1variance} and Lemma \ref{lemma:h2variance} respectively and formed using independent samples.
First calculating the expectation of $(\widehat{aa^\top x})_c$
\begin{equation*}
\begin{aligned}
\E[(\widehat{aa^\top x})_c] = \E[(\widehat{a})_c(\widehat{a^\top x})_c] = \E[(\widehat{a})_c]\E[(\widehat{a^\top x})_c] = aa^\top x
\end{aligned}
\end{equation*}
The above proof uses the fact that $(\widehat{a})_c$ and $(\widehat{a^\top x})_c$ are estimated using independent samples.
Now, calculating the second moment of $\norm{(\widehat{aa^\top x})_c}_2$, we get that
\begin{equation*}
\begin{aligned}
\E\left[\norm{(\widehat{aa^\top x})_c}_2^2\right] &= \E\left[\norm{(\widehat{a})_c(\widehat{a^\top x})_c}_2^2\right]\\
&= \E\left[\norm{(\widehat{a})_c}_2^2\norm{(\widehat{a^\top x})_c}_2^2\right]\\
&= \E\left[\norm{(\widehat{a})_c}_2^2\right]\E\left[(\widehat{aa^\top x})^2_c\right]\\
&\leq \norm{a}_2^2\left(1 + \frac{s(a)}{c}\right)\left((a^\top x)^2 + \frac{1}{c}\norm{\lowerthres{a}{c}}^2_2\norm{x}_2^2\right)
\end{aligned}
\end{equation*}
Now, using Lemma \ref{lemma:sparsityratio}, we know that $\norm{\lowerthres{a}{c}}^2_2 \leq \frac{s(a)}{c}\norm{a}^2_2$\\
Thus, we get that
\begin{equation*}
\begin{aligned}
\E\left[\norm{(\widehat{aa^\top x})_c}_2^2\right] &\leq \norm{a}_2^2\left(1 + \frac{s(a)}{c}\right)\left((a^\top x)^2 + \frac{s(a)}{c^2}\norm{a}^2_2\norm{x}_2^2\right)
\end{aligned}
\end{equation*}
\end{proof}
\hfourvariance*
\begin{proof}
\label{proof:h4variance}
Since, $(\widehat{\mat{A}^\top \mat{A}x})_k = \frac{1}{p_i}(\widehat{a_i{a_i}^\top x})_{c_i}$ with probability
$p_i = \frac{\norm{a_i}_2^2}{M}\left(1 + \frac{s_i}{c_i}\right)$ where $M$ is the normalization constant
 where $(\widehat{a_ia_i^\top x})_{c_i}$ is the estimator of $a_ia_i^\top x$ defined in Lemma~\ref{lemma:h3variance}
  and are formed independently of each other and independently of $i$
  where $s_i = s(a_i)$ and $k$ is some parameter such that $c_i = \sqrt{s_i}k \leq d$.
Calculating the expectation of $(\widehat{\mat{A}^\top \mat{A}x})_k$, we get
\begin{equation*}
\begin{aligned}
\E\left[(\widehat{\mat{A}^\top \mat{A}x})_k\right] &= \E_i \left[\frac{1}{p_i}\E[(\widehat{a_i{a_i}^\top x})_{c_i}]\right]
= \E_i\left[\frac{1}{p_i}a_ia_i^\top x\right]
= \sum_i\frac{p_i}{p_i}a_ia_i^\top x
= \sum_ia_ia_i^\top x
= \mat{A}^\top \mat{A}x
\end{aligned}
\end{equation*}
In the proof above, we used the expectation of $(\widehat{a_i{a_i}^\top x})_{c_i}$ calculated in Lemma~\ref{lemma:h3variance} and
also used that $i$ and $(\widehat{a_i{a_i}^\top x})_{c_i}$ are chosen independently of each other.

Calculating the second moment of $(\widehat{\mat{A}^\top \mat{A}x})_k$, we get
\begin{equation*}
\begin{aligned}
\E\left[\normFull{(\widehat{\mat{A}^\top \mat{A}x})_k}_2^2\right] = \E\left[\normFull{\frac{1}{p_i}(\widehat{a_i{a_i}^\top x})_{c_i}}_2^2\right] &= \E_i\left[\frac{1}{p^2_i}\E\left[\norm{(\widehat{a_i{a_i}^\top x})_{c_i}}_2^2\right]\right]\\
&\leq \E_i\left[\frac{1}{p^2_i}\norm{a_i}_2^2\left(1 + \frac{s_i}{c_i}\right)\left((a_i^\top x)^2 + \frac{s_i}{c_i^2}\norm{a_i}^2_2\norm{x}_2^2\right)\right]\\
&\leq \sum_i\frac{1}{p_i}\norm{a_i}_2^2\left(1 + \frac{s_i}{c_i}\right)\left((a_i^\top x)^2 + \frac{s_i}{c_i^2}\norm{a_i}^2_2\norm{x}_2^2\right)\\
\end{aligned}
\end{equation*}
Putting the value of $p_i = \frac{\norm{a_i}_2^2}{M}\left(1 + \frac{s_i}{c_i}\right)$, we get
\begin{equation*}
\begin{aligned}
\E\left[\normFull{(\widehat{\mat{A}^\top \mat{A}x})_k}_2^2\right] &\leq M\sum_i\left((a_i^\top x)^2 + \frac{s_i}{c_i^2}\norm{a_i}^2_2\norm{x}_2^2\right)\\
\end{aligned}
\end{equation*}
Now, putting the value of $c_i = \sqrt{s_i}k$, we get
\begin{equation*}
\begin{aligned}
\E\left[\normFull{(\widehat{\mat{A}^\top \mat{A}x})_k}_2^2\right] &\leq M\sum_i\left((a_i^\top x)^2 + \frac{1}{k^2}\norm{a_i}^2_2\norm{x}_2^2\right)\\
&\leq M\left(\norm{\mat{A}x}_2^2 + \frac{1}{k^2}\norm{\mat{A}}^2_F\norm{x}_2^2\right)\\
\end{aligned}
\end{equation*}
\end{proof}
\subsection{Application Proofs}
\label{sec:app_proofs}

\subsubsection{Top Eigenvector Computation}
\vareigvecot*
\begin{proof}
Using the unbiasedness of the estimator $(\widehat{\mat{A}^\top \mat{A}x})_k$ from Lemma \ref{lemma:h4variance}, we get
$
\E[\nabla g(x)] = \E[\lambda x - (\widehat{\mat{A}^\top \mat{A}x})_k] = \lambda x - \mat{A}^\top \mat{A}x \\
$
Putting the second moment of $(\widehat{\mat{A}^\top \mat{A}x})_k$ from Lemma \ref{lemma:h4variance} and since in one iteration,
the same $\nabla g(x)$ is picked, the randomness in $\nabla g(x)$ and $\nabla g(x^*)$ is same, therefore, we get
\begin{align}
\E\left[\norm{\nabla g(x) - \nabla g(x^*)}_2^2\right]  &= \E\left[\norm{\lambda(x - x^*) - (\widehat{\mat{A}^\top \mat{A}(x - x^*)})_k}_2^2\right] \nonumber\\
&= \lambda^2\norm{x - x^*}_2^2 - 2\lambda(x - x^*)^\top \E\left[(\widehat{\mat{A}^\top \mat{A}(x - x^*)})_k\right] \nonumber\\
&+ \E\left[\norm{(\widehat{\mat{A}^\top \mat{A}(x - x^*)})_k}_2^2\right]\nonumber
\end{align}
Using the unbiasedness and second moment bound of $(\widehat{\mat{A}^\top \mat{A}x})_k$ from Lemma \ref{lemma:h4variance}, we get
\begin{equation*}
\begin{aligned}
\E\left[\norm{\nabla g(x) - \nabla g(x^*)}_2^2\right] &= \lambda^2(x - x^*)^2 -2\lambda\norm{\mat{A}(x - x^*)}^2_2 \nonumber\\
&+ M\left(\norm{\mat{A}(x - x^*)}_2^2 + \frac{1}{k^2}\norm{\mat{A}}^2_F\norm{x - x^*}_2^2\right) \nonumber\\
&\leq \lambda^2(x - x^*)^2 + M\left(\norm{\mat{A}(x - x^*)}_2^2 + \frac{1}{k^2}\norm{\mat{A}}^2_F\norm{x - x^*}_2^2\right)
\end{aligned}
\end{equation*}
Now, since we want to relate the variance of our gradient estimator the function error to be
used in the SVRG framework, using the strong convexity parameter of the matrix $\mat{B}$, we
get the following:
\begin{equation*}
\begin{aligned}
\E\left[\norm{\nabla g(x) - \nabla g(x^*)}_2^2\right] &\leq \norm{x - x^*}_{\mat{B}}^2\frac{\lambda^2}{\lambda - \lambda_1} + M\norm{x - x^*}_{\mat{B}}^2\left(\frac{\lambda_1}{\lambda - \lambda_1} + \frac{\norm{\mat{A}}_F^2}{(\lambda - \lambda_1)k^2}\right)
\end{aligned}
\end{equation*}
Now, using $k^2 = \sr(\mat{A})$ and rewriting the equation in terms
of problem parameters, $\sr(\mat{A}) = \frac{\norm{\mat{A}}_F^2}{\lambda_1}$ and using $ \left(1 + \frac{\gap}{150}\right)\lambda_1  \leq \lambda \leq 2\lambda_1$, we get
\begin{equation*}
\begin{aligned}
\E\left[\norm{\nabla g(x) - \nabla g(x^*)}_2^2\right] &\leq \norm{x - x^*}_{\mat{B}}^2\frac{4\lambda_1}{\gap} + \frac{2\norm{x - x^*}_{\mat{B}}^2M}{\gap}
\leq 2(f(x) - f(x^*))\left(\frac{4\lambda_1}{\gap} + \frac{2M}{\gap}\right)\\
\end{aligned}
\end{equation*}
It is easy to see that $\lambda_1 \leq M$ and hence, the second term always upper bounds the first term,
thus we get the desired variance bound.
Note from Theorem \ref{thrm:SVRG}, we know that the gradient update is of the following form.
$$x_{k+1} :=  x_k - \eta (\nabla g_{i_k}(x_k) - \nabla g_{i_k}(x_0)) + \eta \nabla f(x_0) ~.$$
Note, the estimator $(\widehat{\mat{A}^\top \mat{A}x})_k$ uses $(\hat{a})_{c_i}$ and $(\widehat{a^\top x})_{c_i}$ estimators internally
which both use $c_i = \sqrt{s_i}k$ coordinates. Sampling to get the estimator can be done in $O(d)$ preprocessing time and $O(1)$ time per sample \cite{kronmal1979alias} and hence $O(c_i)$ time for $c_i$ samples. Also, $\lambda x_k$ can be added to $x_k$ in $O(1)$ time by just maintaining a multiplicative coefficient of
the current iterate and doing the updates accordingly. Hence, our estimator of $\nabla g_{i_k}(x_k) - \nabla g_{i_k}(x_0)$ can be implemented in $O(c_i)$ time when the
ith row is chosen. Furthermore, the dense part $\nabla f(x_0) = \mat{B}x_0 - b$ can be added in $O(1)$ time by separately maintaining
the coefficient of this fixed vector in each $x_k$ and using it as necessary to calculate the $O(c_i)$ coordinates during each iteration.
Consequently, we can bound the total expected time for implementing the iterations by
\begin{equation*}
\begin{aligned}
T &= \sum_i p_ic_i = \sum_i \frac{\norm{a_i}_2^2}{M}\left(1 + \frac{s_i}{c_i}\right)c_i = \sum_i \frac{\norm{a_i}_2^2}{M}(c_i + s_i) = \sum_i \frac{\norm{a_i}_2^2}{M}\left(\sqrt{s_i\sr(\mat{A})} + s_i\right)\\
\end{aligned}
\end{equation*}
\end{proof}

\begin{theorem}[Theorem 16 of \cite{garber2016faster}] \label{thrm:garber}Let us say we have a linear system solver that finds $x$ such that:
$$\E\norm{x - x^*}^2_{\mat{B}} \leq \frac{1}{2}\norm{x_0 - x^*}_{\mat{B}}^2$$
in time $O(T)$ where $f(x) = \frac{1}{2}x^\top \mat{B}x - b^\top x$, $\mat{B} = \lambda \mat{I} - \mat{A}^\top \mat{A}$ and $x_0$ is some initial point. Then we can find an $\epsilon$-approximation $v$ to the top eigenvector of $\mat{A}^\top \mat{A}$ in time $O\left(T \cdot \left(\log^2\left(\frac{d}{\gap}\right) + \log\left(\frac{1}{\epsilon}\right)\right) \right)$ where $ \left(1 + \frac{\gap}{150}\right)\lambda_1  \leq \lambda \leq \left(1 + \frac{\gap}{100}\right)\lambda_1$
\end{theorem}
Theorem 16 of \cite{garber2016faster} states the running time in terms of the running time required for their linear system solver but it can be replaced with any other $\epsilon$-approximate linear system solver.
\begin{theorem} [Linear System Solver Runtime for $\mat{B} = \lambda \mat{I} - \mat{A}^\top \mat{A}$] \label{thrm:eigsvrgruntime}
For a matrix $\mat{B} = \lambda \mat{I} - \mat{A}^\top \mat{A}$, we have an algorithm which returns $x$ such that
$\E\norm{x - x^*}_{\mat{B}}^2 \leq \frac{1}{2}\norm{x^0 - x^*}^2_{\mat{B}}$ in running time $$O\left(\nnz(\mat{A}) + \frac{1}{\gap^2\lambda_1}\sum_i \norm{a_i}_2^2(\sqrt{s_i} + \sqrt{\sr(\mat{A})})\sqrt{s_i}\right)$$ assuming $\lambda_1(1 + \frac{\gap}{150}) \leq \lambda \leq \lambda_1(1 + \frac{\gap}{100})$ where $x^* = \argmin_{x \in \R^d} \frac{1}{2}x^\top \mat{B}x - c^\top x$
\end{theorem}
\begin{proof}
  The problem of solving $\min f(x)$ where $f(x) = \frac{1}{2}x^\top \mat{B}x - c^\top x$ can be solved by using the SVRG framework
  defined in Theorem \ref{thrm:SVRG} with the strong convexity parameter $\mu = \lambda - \lambda_1(\mat{A}^\top \mat{A})$. Using the estimator for $\nabla f(x)$ defined in
   Lemma \ref{lemma:vareigvecot}, we get the corresponding variance defined in Theorem \ref{thrm:SVRG} as $\frac{4M}{\gap}$ i.e. $\sigma^2 = \frac{4M}{\gap}$ where $M$ is as defined in
   Lemma \ref{lemma:vareigvecot}. Therefore, according to Theorem \ref{thrm:SVRG}, we can decrease the error by a constant factor in total
   number of iterations $O(\frac{4M}{\gap(\lambda - \lambda_1)})$. The expected time taken per iteration is $T = \sum_i \frac{\norm{a_i}_2^2}{M}\left(s_i + \sqrt{\sr(\mat{A})s_i}\right)$
    as defined in Lemma \ref{lemma:vareigvecot}. Now, we can argue that the total expected time taken per iteration would be $T\frac{\sigma^2}{\lambda - \lambda_1}$.
    Therefore, the total time expected taken to decrease the error by a constant factor would be
    $O\left(\nnz(\mat{A}) + \frac{1}{\lambda - \lambda_1}\frac{M}{\gap}\sum_i \frac{\norm{a_i}_2^2}{M}\left(s_i + \sqrt{s_i\sr(\mat{A})}\right)\right)$. \\
  Simplifying, we get
  \begin{equation*}
  \begin{aligned}
  \frac{1}{\lambda - \lambda_1}\frac{M}{\gap}\sum_i \frac{\norm{a_i}_2^2}{M}\left(s_i + \sqrt{s_isr(\mat{A})}\right)
  = \frac{1}{\gap^2}\sum_i \frac{\norm{a_i}_2^2}{\lambda_1}\left(\sqrt{s_i} + \sqrt{sr(\mat{A})}\right)\sqrt{s_i}\\
  \end{aligned}
  \end{equation*}
We can argue this holds with constant probability by using Markov's inequality.
\end{proof}
\eigruntime*
\begin{proof}
  This follows directly from combining Theorem \ref{thrm:garber} and Theorem \ref{thrm:eigsvrgruntime}.
\end{proof}
\acceigruntime*
\begin{proof}
The total running time for solving the regularized linear system in $\mat{B} + \gamma \mat{I} = (\lambda + \gamma)I - \mat{A}^\top \mat{A}$ upto constant accuracy is $O\left(\nnz(\mat{A}) + \frac{2\lambda_1(\mat{A}^\top \mat{A})}{(\lambda + \gamma - \lambda_1)^2}\sum_i \norm{a_i}_2^2\left(\sqrt{s_i} + \sqrt{\sr(\mat{A})}\right)\sqrt{s_i}\right)$ by Theorem~\ref{thrm:eigsvrgruntime}.
  Hence, the total running time for solving the unregularized linear system in $\mat{B}$ according to Theorem~\ref{thrm:acceleration} will be \\$O\left(\left(\nnz(\mat{A}) + \frac{2\lambda_1}{(\lambda + \gamma - \lambda_1)^2}\sum_i \norm{a_i}_2^2\left(\sqrt{s_i} + \sqrt{\sr(\mat{A})}\right)\sqrt{s_i}\right)\log\left(\frac{2\gamma}{\lambda - \lambda_1}\right)\sqrt{\frac{\gamma}{\lambda - \lambda_1}}\right)$
   assuming $\gamma \geq 2(\lambda - \lambda_1)$ by the assumption of the theorem.

  Balancing the two terms, we get that $$\gamma = \sqrt{\frac{2\lambda_1}{\nnz(\mat{A})}\sum_i \norm{a_i}_2^2\left(\sqrt{s_i} + \sqrt{\sr(\mat{A})}\right)\sqrt{s_i}}$$
  Putting this in the total runtime and using $\frac{\gamma}{\lambda - \lambda_1} \leq \sqrt{\frac{d\norm{\mat{A}}_F^2}{\lambda_1\nnz(\mat{A})\gap^2}}$, we get a total runtime of $$\tilde{O}\left(\nnz(\mat{A}) + \nnz(\mat{A})\left({\frac{\lambda_1}{\nnz(\mat{A})(\lambda - \lambda_1)^2}\sum_i \norm{a_i}_2^2\left(\sqrt{s_i} + \sqrt{\sr(\mat{A})}\right)\sqrt{s_i})}\right)^{\frac{1}{4}}\right)$$ where $\tilde{O}$ hides a factor of $\log\left(\frac{d\norm{\mat{A}}_F^2}{\lambda_1\nnz(\mat{A})\gap^2}\right)$.

  Since $\frac{\norm{\mat{A}}_F^2}{\lambda_1} \leq d$ and $\nnz(\mat{A}) \geq 1$, we get a running time of
  $$O\left(\left(\nnz(\mat{A}) + \frac{\nnz(\mat{A})^{\frac{3}{4}}}{\sqrt{\gap}}\left({\sum_i \frac{\norm{a_i}_2^2}{\lambda_1}\left(\sqrt{s_i} + \sqrt{\sr(\mat{A})}\right)\sqrt{s_i})}\right)^{\frac{1}{4}}\right)\log\left(\frac{d}{\gap}\right)\right)$$
 for solving a linear system
  in $B$. Then, the final running time is obtained from combining Theorem \ref{thrm:garber} with the time for the linear system solver in $\mat{B}$ obtained above.
\end{proof}
\subsubsection{Regression}

\varlinreg*
\begin{proof}
Since $(\widehat{\mat{A}^\top \mat{A}x})_k$ is an unbiased estimate from Lemma~\ref{lemma:h4variance}, we get
\(
\E[\nabla g(x)]
= \E[(\widehat{\mat{A}^\top \mat{A}x})_k]
= \mat{A}^\top \mat{A}x
\)

To calculate $\E\left[\norm{\nabla g(x) - \nabla g(x^*)}_2^2\right]$,
using the second moment of $(\widehat{\mat{A}^\top \mat{A}x})_k$ from Lemma \ref{lemma:h4variance} and since in one iteration,
the same $\nabla g(x)$ is picked, the randomness in $\nabla g(x)$ and $\nabla g(x^*)$ is same, we get
\begin{align}
\E\left[\norm{\nabla g(x) - \nabla g(x^*)}_2^2\right] &= \E\left[\norm{(\widehat{\mat{A}^\top \mat{A}(x - x^*)})_k}_2^2\right]\nonumber\\
&\leq M\left(\norm{\mat{A}(x - x^*)}_2^2 + \frac{1}{k^2}\norm{\mat{A}}^2_F\norm{x - x^*}_2^2\right)\nonumber
\end{align}
Putting the value of $k = \sqrt{\kappa
}$, and using strong convexity, $\norm{\mat{A}(x - x^*)}_2^2 \geq \mu\norm{x - x^*}_2^2$ we get that
\begin{align*}
\E\left[\norm{\nabla g(x) - \nabla g(x^*)}_2^2\right]
&\leq M\left(\norm{\mat{A}(x - x^*)}_2^2 + \norm{\mat{A}(x - x^*)}_2^2\right)
%\\
%&
\leq 2M\norm{\mat{A}(x - x^*)}_2^2
% \\
\end{align*}
Note from Theorem \ref{thrm:SVRG}, we know that the gradient update is of the following form.
\abovedisplayshortskip=5pt
\[
x_{k+1} :=  x_k - \eta (\nabla g_{i_k}(x_k) - \nabla g_{i_k}(x_0)) + \eta \nabla f(x_0)
\]
Note, the estimator $(\widehat{\mat{A}^\top \mat{A}x})_k$ uses $(\hat{a})_{c_i}$ and $(\widehat{a^\top x})_{c_i}$ estimators internally
which both use $c_i = \sqrt{s_i}k$ coordinates. Sampling to get the estimator can be done in $O(d)$ preprocessing time and $O(1)$ time
per sample \cite{kronmal1979alias} and hence $O(c_i)$ time for $c_i$ samples. Hence, our estimator of $\nabla g_i(x_k) - \nabla g_i(x_0)$ can be implemented in $O(c_i)$ time when the
ith row is chosen. Furthermore, the dense part $\nabla f(x_0) = \mat{A}^\top (\mat{A}x_0 - b)$ can be added in $O(1)$ time by separately maintaining
the coefficient of this fixed vector in each $x_k$ and using it as necessary to calculate the $O(c_i)$ coordinates during each iteration. Consequently, we can bound the total expected time for implementing the iterations by
\begin{align*}
T
&= \sum_{i \in [n]} p_ic_i
= \sum_{i \in [n]} \frac{\norm{a_i}_2^2}{M}\left(1 + \frac{s_i}{c_i}\right)c_i
= \sum_{i \in [n]} \frac{\norm{a_i}_2^2}{M}(c_i + s_i)
= \sum_{i \in [n]} \frac{\norm{a_i}_2^2}{M}\left(\sqrt{\kappa} + \sqrt{s_i}\right)\sqrt{s_i}
\end{align*}
Now, we know that $s_i \leq d$ and $\kappa \geq d$, hence,
the first term in the above expression always dominates. hence, we get the desired
bound on $T$.
\end{proof}

\runtimeregress*
\begin{proof}
The problem of solving regression $\min f(x)$ where $f(x) = \frac{1}{2}\norm{\mat{A}x - b}_2^2$ can be solved by using the SVRG framework
defined in Theorem \ref{thrm:SVRG} with the strong convexity parameter $\mu = \lambda_d(\mat{A}^\top \mat{A})$. Using the estimator for $\nabla f(x)$ defined in
 Lemma \ref{lemma:varlinreg}, we get the corresponding variance defined in Theorem \ref{thrm:SVRG} as $M$ i.e. $\sigma^2 = M$ where $M$ is as defined in
 Lemma \ref{lemma:varlinreg}. Therefore, according to Theorem \ref{thrm:SVRG}, we can decrease the error by a constant factor in total
 number of iterations $O(\frac{M}{\mu})$. The expected time taken per iteration is $T = \frac{\sqrt{\kappa}}{M}\sum_{i \in [n]} \norm{a_i}_2^2\sqrt{s_i}$
  as defined in Lemma \ref{lemma:varlinreg}. Now, since we know the number of iterations is $O(\frac{\sigma^2}{\mu})$ and expected time per iteration is
  $O(T)$,
  the total expected running time will be $O(T\frac{\sigma^2}{\mu}) = O\left(\nnz(\mat{A}) + \sqrt{\kappa}\sum_{i \in [n]} \frac{\norm{a_i}_2^2}{\mu}\sqrt{s_i}\right)$ to decrease the error by constant
multiplicative factor. Thus, the total expected time
  taken to get an $\epsilon$-approximate solution to the problem would be
 $O\left(\left(\nnz(\mat{A}) + \sqrt{\kappa}\sum_{i \in [n]} \frac{\norm{a_i}_2^2}{\mu}\sqrt{s_i}\right)\log\left(\frac{1}{\epsilon}\right)\right)$.
 We can argue this holds with constant probability using Markov's inequality.
\end{proof}
\accruntimeregress*
\begin{proof}
  \label{proof:acc_runtime_regress}
Solving a regularized least squares problem i.e. $\min_{x}\norm{\mat{A}x - b}_2^2 + \lambda\norm{x - x_0}_2^2$ is equivalent to solving
a problem with a modified matrix $\mat{\tilde{A}}$ with $n + d$ rows with the last $d$ rows having sparsity of $1$ and rows $\tilde{a}_i = a_i$ if $i \leq n$ and $\lambda \hat{e}_{i-n}$ otherwise,
$\tilde{s}_i = s_i$ if $i \leq n$ and $1$ otherwise and $\tilde{\mu} = \mu + \lambda$ and therefore,
by Theorem~\ref{thm:runtime_regress}, the running time for solving the regularized regression up to constant accuracy will be
\[O\left(\nnz(\mat{\tilde{A}}) + \sqrt{\frac{\norm{\mat{\tilde{A}}}^2_F}{\mu + \lambda}}\sum_{i \in [n+d]} \frac{\norm{\tilde{a_i}}^2_2}{\mu + \lambda}\sqrt{\tilde{s}_i}\right)
\]
which is equal to
\[O\left(\nnz(\mat{A}) + d + \sqrt{\frac{\norm{\mat{A}}^2_F + d\lambda^2}{\mu + \lambda}}\sum_{i \in [n]} \frac{\norm{a_i}^2_2}{\mu + \lambda}\sqrt{s_i} + \sqrt{\frac{\norm{\mat{A}}^2_F + d\lambda^2}{\mu + \lambda}}\frac{d\lambda^2}{\mu + \lambda}\right) ~.
\]
Thus, the total running time for solving the unregularized problem will be, by Theorem~\ref{thrm:acceleration}
$$\tilde{O}\left(\left(\nnz(\mat{A}) + d + \sqrt{\frac{\norm{\mat{A}}^2_F + d\lambda^2}{\mu + \lambda}}\sum_{i \in [n]} \frac{\norm{a_i}^2_2}{\mu + \lambda}\sqrt{s_i} + \sqrt{\frac{\norm{\mat{A}}^2_F + d\lambda^2}{\mu + \lambda}}\frac{d\lambda^2}{\mu + \lambda}\right)\sqrt{\frac{\lambda}{\mu}}\right)$$
where $\tilde{O}$ hides a factor of $\log\left(\frac{\lambda + 2\mu}{\mu}\right)$.
Since $\lambda > 2\mu$ by the assumption of Theorem~\ref{thrm:acceleration}, we get $\lambda + \mu = O(\lambda)$, thus the total running time becomes
$$O\left(\left(\nnz(\mat{A}) + d + \sqrt{\frac{\norm{\mat{A}}^2_F + d\lambda^2}{\lambda}}\sum_{i \in [n]} \frac{\norm{a_i}^2_2}{\lambda}\sqrt{s_i} + \sqrt{\frac{\norm{\mat{A}}^2_F + d\lambda^2}{ \lambda}}\frac{d\lambda^2}{\lambda}\right)\sqrt{\frac{\lambda}{\mu}}\log\left(\frac{\lambda}{\mu}\right)\right)$$
Assuming $nnz(\mat{A}) \geq d$ and $\lambda^2 < \frac{\norm{\mat{A}}^2_F}{d}$ since $\lambda$
should be less than the smoothness parameter of the problem.
The running time becomes
$$O\left(\left(\nnz(\mat{A}) + \sqrt{\frac{\norm{\mat{A}}^2_F}{\lambda}}\sum_{i \in [n]} \frac{\norm{a_i}^2_2}{\lambda}\sqrt{s_i} + \sqrt{\frac{\norm{\mat{A}}^2_F}{ \lambda}}\frac{d\lambda^2}{\lambda}\right)\sqrt{\frac{\lambda}{\mu}}\log\left(\frac{\lambda}{\mu}\right)\right)$$
Since, $\lambda^2 \leq\frac{\norm{\mat{A}}^2_F}{d} \leq \frac{1}{d}\sum_{i \in [n]}\norm{a_i}_2^2\sqrt{s_i}$ as $s_i \geq 1$,
hence, we get that the running time becomes
$$O\left(\left(\nnz(\mat{A}) + \sqrt{\frac{\norm{\mat{A}}^2_F}{\lambda}}\sum_{i \in [n]} \frac{\norm{a_i}^2_2}{\lambda}\sqrt{s_i}\right)\sqrt{\frac{\lambda}{\mu}}\log\left(\frac{\lambda}{\mu}\right)\right)$$
Balancing the two terms, we get the value of $\lambda = {\left({\frac{\norm{\mat{A}}_F}{\nnz(\mat{A})}}\sum_{i \in [n]} \norm{a_i}^2_2\sqrt{s_i}\right)}^{\frac{2}{3}}$ which also satisfies our assumption on $\lambda$, hence the total running time becomes $$O\left(\nnz(\mat{A}) + \frac{\nnz(\mat{A})}{\sqrt{\mu}}{\left({\frac{\norm{\mat{A}}_F}{\nnz(\mat{A})}}\sum_{i \in [n]} \norm{a_i}^2_2\sqrt{s_i}\right)}^{\frac{1}{3}}\log\left({\frac{\norm{\mat{A}}_F}{\mu^{\frac{3}{2}}\nnz(\mat{A})}}\sum_{i \in [n]} \norm{a_i}^2_2\sqrt{s_i}\right)\right)$$ where
an additional $\nnz(\mat{A})$ is there because we need to look at all the entries of the matrix once.
Thus, the running time for solving the system upto $\epsilon$ accuracy will be $$O\left(\left(\nnz(\mat{A}) + \nnz(\mat{A})^{\frac{2}{3}}{\kappa}^{\frac{1}{6}}{\left(\sum_{i \in [n]} \frac{\norm{a_i}^2_2}{\mu}\sqrt{s_i}\right)}^{\frac{1}{3}}\log\left(\kappa   \right)\right)\log\left(\frac{1}{\epsilon}\right)\right)$$
\end{proof}
\section{Entrywise Sampling}
\label{entrywisesampling}
In this section, we compute what bounds we get by first doing entrywise sampling on matrix $\mat{A}$ to get $\mat{\tilde{A}}$ and then running
regression on $\mat{\tilde{A}}$.
Let us say we do entrywise sampling on the matrix $\mat{A} \in \R^{n \times d}$ to obtain a matrix $\mat{\tilde{A}} \in \R^{n \times d}$ with $s$ non-zero entries such that $\norm{\mat{A} - \mat{\tilde{A}}}_2 \leq \epsilon$.\\
Then, we can write $\mat{\tilde{A}} = \mat{A} + \mat{E}$ where $\mat{E} \in \R^{n \times d}$ where $\norm{\mat{E}}_2 \leq \epsilon$ \\
To get a bound on $\norm{\mat{A}^\top \mat{A} - \mat{\tilde{A}}^\top \mat{\tilde{A}}}_2$
\begin{align}
\norm{\mat{A}^\top \mat{A} - \mat{\tilde{A}}^\top \mat{\tilde{A}}}_2 &= \norm{\mat{A}^\top \mat{A} - {(\mat{A} + \mat{E})}^\top (\mat{A} + \mat{E})}_2 \nonumber\\
&= \norm{-\mat{A}^\top \mat{E} - \mat{E}^\top \mat{A} - \mat{E}^\top \mat{E}}_2 \nonumber\\
&\leq 2\sigma_{\max}(\mat{A})\epsilon + \epsilon^2 \nonumber
\end{align}
For $\epsilon \leq \frac{\lambda_{\min}(\mat{A}^\top \mat{A})\epsilon'}{2\sigma_{\max}(\mat{A})}$, we get
$$\norm{\mat{A}^\top \mat{A} - \mat{\tilde{A}}^\top \mat{\tilde{A}}}_2 \leq \lambda_{\min}(\mat{A}^\top \mat{A})\epsilon'$$
and thus, we get that $\mat{\tilde{A}}^\top \mat{\tilde{A}}$ is a spectral approximation to $\mat{A}^{\top}\mat{A}$ i.e.
$$(1 - \epsilon')\mat{A}^\top \mat{A} \leq \mat{\tilde{A}}^\top \mat{\tilde{A}} \leq (1 + \epsilon')\mat{A}^\top \mat{A}$$

Thus, we can solve a linear system in $\mat{A}^\top \mat{A}$ to get $\delta$ multiplicative accuracy by solving $\log(\frac{1}{\delta})$ linear systems in $\mat{\tilde{A}}^\top \mat{\tilde{A}}$ upto constant accuracy and hence, the total running time will be $(\nnz(\mat{\tilde{A}}) + \frac{\norm{\mat{\tilde{A}}}^2_Fs'}{\lambda_{\min}(\mat{\tilde{A}}^\top \mat{\tilde{A}})})\log(\frac{1}{\delta})$ where $s'$ is the number of entries per row of $\mat{\tilde{A}}$ i.e. we can find $x$ such that $\norm{x - x^{*}}_{\mat{A}^\top \mat{A}} \leq \delta\norm{x_0 - x^{*}}_{\mat{A}^\top \mat{A}}$ where $\mat{A}^\top \mat{A}x^{*} = c$.

Assuming uniform sparsity which is the best case for this appraoch and might not be true in general, we get the following running times by instantiating
the above running time with different entry wise sampling results.

Using the results in \cite{achlioptas2007fast} we get, $s' = O(\frac{||A||_F^2}{\epsilon^2})$ where $\epsilon = \frac{\lambda_{\min}(\mat{A}^\top \mat{A})\epsilon'}{2\sigma_{\max}(\mat{A})}$
and hence $s' = \frac{||A||_F^2\sigma^2_{\max}(\mat{A})}{\lambda^2_{\min}(\mat{A}^\top \mat{A}){\epsilon'}^2}$ and hence a total running time of
$\tilde{O}\left(\nnz(\mat{A}) + \frac{\norm{\mat{A}}^4_F\lambda_{\max}(\mat{A}^{\top}\mat{A})}{\lambda^3_{\min}(\mat{A}^\top \mat{A})})\right)$

Using the results in \cite{achlioptas2013near} we get, $s' = O(\frac{\sum_i \norm{\mat{A}_{(i)}}_1^2}{n\epsilon^2})$ and hence we get a total running time of\\
$\tilde{O}\left(\nnz(\mat{A}) + \frac{\sum_i \norm{\mat{A}_{(i)}}_1^2\norm{\mat{A}}^2_F\lambda_{\max}(\mat{A}^{\top}\mat{A})}{n\lambda^3_{\min}(\mat{A}^\top \mat{A})})\right)$ or
$\tilde{O}\left(\nnz(\mat{A}) + \frac{\sum_i s_i\norm{a_i}_2^2\norm{\mat{A}}^2_F\lambda_{\max}(\mat{A}^{\top}\mat{A})}{n\lambda^3_{\min}(\mat{A}^\top \mat{A})})\right)$

Using the results in \cite{arora2006fast} we get, $s' = \sum_{ij}\frac{|\mat{A}_{ij}|}{\sqrt{n}\epsilon}$ and hence, $s' = \frac{\sigma_{\max}(\mat{A})}{\lambda_{\min}(\mat{A}^\top \mat{A})}\sum_{ij}\frac{|\mat{A}_{ij}|}{\sqrt{n}\epsilon}$
and hence a total running time of $\tilde{O}\left( \nnz(\mat{A}) + \frac{\sum_{ij}|\mat{A}_{ij}|\norm{\mat{A}}^2_F\sigma_{\max}(\mat{A})}{\sqrt{n}\lambda^2_{\min}(\mat{A}^\top \mat{A})} \right)$ or
$\tilde{O}\left( \nnz(\mat{A}) + \frac{\sum_{i}\sqrt{s_{i}}\norm{a_i}_2\norm{\mat{A}}^2_F\sqrt{\lambda_{\max}(\mat{A}^{\top}\mat{A})}}{\sqrt{n}\lambda^2_{\min}(\mat{A}^\top \mat{A})} \right)$
\begin{table}[h]
\begin{center}
\caption{Previous results in entry wise matrix sparsification. Given a matrix $\mat{A} \in \R^{n \times n}$, we want to have a sparse matrix $\mat{\tilde{A}} \in \R^{n \times n}$ satisfying $\norm{\mat{A} - \mat{\tilde{A}}}_2 \leq \epsilon$. The first column indicates the number of entries in $\mat{\tilde{A}}$ (in expectation or exact).
  Note that this is not a precise treatment of entrywise sampling results since some results grouped together in the first row have different success probabilities and some results
  also depends on the ratio of the maximum and minimum entries in the matrix but this is the lower bound and we ignore details for simplicity since this suffices for our comparison. \label{tab:Results}}
{\footnotesize
{\tabulinesep=1mm
\begin{tabu}{ |c|c| }
\hline
\multicolumn{2}{ |c| }{\bf Previous entry wise sampling results}  \\
\hline
{\bf Sparsity of $\mat{\tilde{A}}$ in $\tilde{O}$} & {\bf Citation} \\
\hline
{$n\frac{\norm{\mat{A}}_F^2}{\epsilon^2}$} &  {\cite{achlioptas2007fast, gittens2009error, nguyen2009matrix, nguyen2010tensor, drineas2011note, kundu2014note}}  \\
\hline
{$ \sqrt{n}\sum_{ij}\frac{|\mat{A}_{ij}|}{\epsilon}$} & \cite{arora2006fast} \\
\hline
{$\sum_i \frac{||a_{i}||_1^2}{\epsilon^2} + \sqrt{\frac{||\mat{A}||_1^2}{\epsilon^2}   }$} & {\cite{achlioptas2013near}} \\
\hline
\end{tabu}
}
}
\end{center}
\end{table}

\end{document}